\documentclass{article}

\usepackage{microtype}
\usepackage{graphicx}
\usepackage{subcaption}
\usepackage{booktabs} 
\usepackage{float}

\usepackage{hyperref}
\usepackage{capt-of}
\usepackage{comment}

\usepackage[accepted]{icml2026}

\usepackage{amsmath}
\usepackage{amssymb}
\usepackage{mathtools}
\usepackage{amsthm}

\usepackage[capitalize,noabbrev]{cleveref}

\theoremstyle{plain}
\newtheorem{theorem}{Theorem}[section]
\newtheorem{proposition}[theorem]{Proposition}
\newtheorem{lemma}[theorem]{Lemma}
\newtheorem{corollary}[theorem]{Corollary}
\theoremstyle{definition}

\theoremstyle{remark}
\newtheorem{remark}[theorem]{Remark}
\newtheorem{example}{Example}

\usepackage[textsize=tiny]{todonotes}

\icmltitlerunning{Block-wise Codeword Embedding for Reliable Multi-bit Text Watermarking}

\begin{document}

\twocolumn[
  \icmltitle{Block-wise Codeword Embedding for \\ 
  Reliable Multi-bit Text Watermarking}




  \begin{icmlauthorlist}
    \icmlauthor{Joeun Kim}{dgistai}
    \icmlauthor{HoEun Kim}{dgistai}
    \icmlauthor{Dongsup Jin}{ulsan}
    \icmlauthor{Young-Sik Kim}{dgistai,dgisteecs}
  \end{icmlauthorlist}
    \icmlaffiliation{dgistai}{Department of AI, DGIST, Daegu, Republic of Korea}
    \icmlaffiliation{dgisteecs}{Department of EECS, DGIST, Daegu, Republic of Korea}
    \icmlaffiliation{ulsan}{Department of ICT Convergence, University of Ulsan, Ulsan, Republic of Korea}
    \icmlcorrespondingauthor{Young-Sik Kim}{ysk@dgist.ac.kr}
  \icmlkeywords{Machine Learning, ICML}
  \vskip 0.3in
]



\printAffiliationsAndNotice{}  

\begin{abstract}
Recent multi-bit watermarking methods for large language models (LLMs) prioritize capacity over reliability, often conflating decoding with detection. Our analysis reveals that existing ECC-based extractors suffer from catastrophic false positive rates (FPR), and applying rejection thresholds merely collapses detection sensitivity (TPR) to random guessing. To resolve this structural limitation, we propose \textbf{BREW} (Block-wise Reliable Embedding for Watermarking), a framework shifting the paradigm to \emph{designated verification}. BREW employs a two-stage mechanism: (i) \textbf{blind message estimation} via independent block voting, followed by (ii) \textbf{window-shifting verification} that rigorously validates the payload against local edits. Experiments demonstrate that BREW achieves a TPR of 0.965 with an FPR of 0.02 under 10\% synonym substitution, demonstrating that the high-FPR issue is not an inherent trade-off of multi-bit watermarking, but a solvable structural flaw of prior decoding-centric designs. Our framework is model-agnostic and theoretically grounded, providing a scalable solution for reliable forensic deployment.
\end{abstract}

\section{Introduction}
Large language models (LLMs) have transformed content generation across creative, professional, and scientific domains, yet raise critical concerns about provenance and potential misuse for deceptive content \cite{solaiman2019release,bender2021dangers}. Reliably distinguishing human-authored from AI-generated text has become essential for academic integrity, journalism, legal proceedings, and platform governance \cite{mitchell2023detectgpt,gehrmann2019gltr}.
Text watermarking addresses this challenge by embedding imperceptible data into AI-generated content during generation \cite{kirchenbauer2023watermark}. Unlike post-hoc detection methods relying on statistical artifacts \cite{mitchell2023detectgpt,su2023detectllm}, watermarking provides stronger origin guarantees while preserving fluency and style.
Watermarking approaches divide into zero-bit (checking watermark presence) and multi-bit (encoding extractable metadata). The green/red partition strategy of \cite{kirchenbauer2023watermark} biases generation toward a keyed ``green'' vocabulary subset. Recent multi-bit methods augment partitioning with error-correcting codes (ECCs) to embed message bits. \cite{qu2025provably} encodes payloads with Reed--Solomon codes, while \cite{chao2024rbc} uses LDPC codes with sliding windows for short texts.

Despite strong extraction capabilities, current multi-bit watermarking lacks explicit false positive control \cite{fu2025multi}. For instance, under 10\% synonym substitution, a non-ECC approach \cite{yoo2024advancing} shows substantial false positives (FPR $\approx 0.5$) despite perfect recall. Similarly, our reproduction of the ECC-based method by \cite{qu2025provably} reveals an unacceptable FPR ($>0.9$) alongside high TPR ($\approx 0.97$), frequently misclassifying unwatermarked text.
We introduce \textbf{B}lock-wise \textbf{R}eliable \textbf{E}mbedding for \textbf{W}atermarking (BREW), a block-wise multi-bit watermarking framework designed to control false positives while preserving high detection sensitivity. Specifically, BREW (i) embeds \emph{complete codewords in independent blocks} to localize errors and prevent cascade failures under edits, and (ii) deploys a \emph{window-shifting detector} that systematically realigns and recovers codewords after insertion/deletion-induced desynchronization. Crucially, detection verifies that a recovered codeword equals the \emph{designated} codeword that was actually embedded in that block, thereby suppressing spurious matches that inflate FPR. This design achieves both a high TPR and a significantly lower FPR compared to previous multi-bit methods, making it more suitable for real-world forensic applications. The framework is \emph{code-agnostic}: while we instantiate with BCH codes for efficiency and clarity, the design extends to RS/LDPC codes, enabling adaptation to application-specific error patterns.

As a result, BREW closes a critical reliability gap in prior multi-bit watermarking. On 200-token texts under 10\% synonym substitution, BREW achieves strong detection performance (TPR $=0.965$) while maintaining a low false positive rate (FPR $=0.02$), in sharp contrast to both earlier non-ECC approaches with substantial false positives and recent ECC-based methods whose FPR exceeds $0.9$.

\paragraph{Contributions} This work makes the following key contributions:
\begin{enumerate}
    \item \textbf{Paradigm Shift to Verification-Centric Detection}: We identify that high FPR stems from conflating \emph{decoding} with \emph{detection}. BREW decouples these via designated verification. Unlike ``steel-manned'' baselines that rely on rejection thresholds (suffering severe recall loss), our approach fundamentally solves the false detection problem.
    \item \textbf{Reliable Short-Payload Transmission}: We focus on the guaranteed delivery of critical payloads in adversarial settings. Our distributed architecture ensures robustness against localized attacks where traditional methods fail.
    \item \textbf{Incremental Detection Framework}: Watermark evidence accumulates from independent blocks, enabling graduated confidence assessment and partial recovery even when portions of the text are corrupted.
    \item \textbf{Theory for Reliability}: We provide finite-sample bounds that rigorously control false positives and characterize detection power under realistic noise models.
    \item \textbf{Comprehensive Validation}: Experiments across multiple datasets and models (OPT, LLaMA, Mistral) demonstrate state-of-the-art TPR--FPR trade-offs. BREW achieves an FPR of 0.02 under 10\% substitution, whereas prior methods exceed 0.90.
\end{enumerate}



\section{Related Work}
\label{sec:related}

\begin{table*}[t]
  \caption{Comparison of zero-bit text watermarking methods.}
  \label{tab:zero_bit_methods}
  \centering
  \small
  \renewcommand{\arraystretch}{1.2}
  \begin{tabular}{p{3.2cm}p{1.6cm}p{5.6cm}p{4.8cm}}
    \toprule
    Method & ECC & Key Idea & Main Limitations \\
    \midrule
    KGW~\cite{kirchenbauer2023watermark} & No &
    Green/red token partition; binomial hypothesis test &
    Vulnerable to paraphrasing; weak signal in short texts \\

    DiPmark~\cite{wu2024dipmark} & No &
    Distribution-preserving biasing; minimal quality loss &
    Reduced detection power \\

    \cite{zhao2024provable} & No &
    Unigram statistics; provable robustness bounds &
    Limited expressiveness beyond unigrams \\

    NS-Watermark~\cite{takezawa2023necessary} & No &
    Detectability theory; exponential reweighting &
    Lacks practical robustness mechanisms \\

    \cite{christ2024pseudorandom} & Yes &
    Pseudorandom ECC; hidden hypothesis testing &
    High computational cost; zero payload capacity \\
    \bottomrule
  \end{tabular}
\end{table*}

\begin{table*}[t]
  \caption{Comparison of multi-bit text watermarking methods.}
  \label{tab:multi_bit_methods}
  \centering
  \small
  \renewcommand{\arraystretch}{1.2}
  \begin{tabular}{p{3.2cm}p{1.6cm}p{5.6cm}p{4.8cm}}
    \toprule
    Method & ECC & Key Idea & Main Limitations \\
    \midrule
    MPAC~\cite{yoo2024advancing} & No &
    Pseudorandom position allocation; zero-bit watermark composition &
    Low extraction accuracy (49.2\% @ 32 bits); no calibrated detection \\

    \cite{qu2025provably} & Yes (RS) &
    Reed--Solomon message encoding; segment-level voting &
    Very high FPR ($\approx 0.9$) under insertion/deletion \\

    RBC~\cite{chao2024rbc} & Yes (LDPC) &
    Sliding-window decoding; adaptive biasing &
    Elevated FPR risk; complex decoding and tuning \\
    \bottomrule
  \end{tabular}
\end{table*}

Text watermarking for LLMs has rapidly diversified alongside model capabilities and deployment contexts. We organize prior work by \emph{detection objective}: (i) \emph{zero-bit} watermarking, which only tests for the presence of a watermark, and (ii) \emph{multi-bit} watermarking, which embeds and extracts a payload. This lens clarifies robustness requirements (synchronization, error tolerance) and evaluation protocols, and it better reflects recent cryptographic developments, including zero-bit constructions based on pseudorandom error-correcting codes.

\subsection{Zero-Bit Watermarking}
A canonical approach is the keyed green/red partition of KGW \cite{kirchenbauer2023watermark}, which biases generation toward a secret per-token green set and applies a binomial-style hypothesis test at detection. Variants preserve the model output distribution to improve text quality, as in DiPmark \cite{wu2024dipmark}, or provide robustness under bounded edits by operating on unigram statistics \cite{zhao2024provable}. From a theoretical perspective, NS-Watermark \cite{takezawa2023necessary} further sharpens detectability by deriving necessary and sufficient conditions and realizing them via exponential reweighting. From a cryptographic angle, \cite{christ2024pseudorandom} constructs pseudorandom error-correcting codes whose local neighborhoods are computationally indistinguishable from random, enabling hidden presence tests with constant error. Despite their efficiency, most zero-bit schemes rely on aggregate frequency signals and lack explicit synchronization, making them vulnerable to paraphrasing, translation, or token-level desynchronization, especially in short texts.

\subsection{Multi-Bit Watermarking}
Multi-bit watermarking seeks to embed a payload that can be \emph{decoded}. Two broad families appear.
\paragraph{(a) Non-ECC Multi-Bit Ideas.} 
MPAC \cite{yoo2024advancing} uses invariant features (keywords/syntax) for robustness, but suffers allocation imbalance and low accuracy on longer messages (49.2\% match rate for 32-bit).

\paragraph{(b) ECC-Based Message Encoding.}
\cite{qu2025provably} pioneers \emph{ECC-based message encoding}, which encodes the payload with Reed--Solomon (RS), distributes symbols via pseudorandom segments, and decodes by cracking noisy segment votes to the nearest codeword. RBC \cite{chao2024rbc} extends this line with LDPC and sliding windows, reporting strong performance on short texts through adaptive biasing and sophisticated decoding.

\paragraph{Limitations.}
ECC-based methods often behave like \emph{message extractors}, not calibrated detectors: nearest-codeword decoding maps even unwatermarked text to valid codewords, driving FPR high—particularly under insertions/deletions or synonym edits. \cite{fu2025multi} formalizes this \emph{false detection problem}: conflating detection with identification effectively enlarges key capacity and degrades reliability.

\subsection{Attacks and Evaluation Protocols}
Attacks include (i) \textbf{substitutions}, including synonym replacement, back-translation, and model-based paraphrasing \cite{morris2020textattack, wieting2018paranmt, krishna2023paraphrasing}, (ii) \textbf{insertions/deletions} that break token–bit alignment, and (iii) \textbf{semantic rewrites} that alter surface form while preserving meaning \cite{wolff2020attacking}. While recent frameworks standardize protocols and metrics \cite{kuditipudi2023robust}, insertion/deletion scenarios remain underexplored. Prior pseudorandom embedding strategies \cite{yoo2024advancing, qu2025provably} mitigate—but do not resolve—synchronization, and segment-level voting can yield unacceptably high FPR on unwatermarked text.

\subsection{Positioning of Our Work: From Extraction to Verification}
Existing ECC-based methods \cite{qu2025provably, chao2024rbc} primarily function as \emph{message extractors}: they aggressively map noisy signals to the nearest valid codeword, inherently assuming the presence of a watermark. Even when ``steel-manned'' with optimal rejection thresholds (as analyzed in Section~\ref{sec:experiments}), these extractors cannot escape a linear trade-off between TPR and FPR (approaching random guessing) under desynchronization attacks.

We position BREW as a \textbf{verification-centric framework} designed for the \emph{reliable transmission of short payloads} in adversarial environments. Unlike methods that strive for high-capacity transmission (often sacrificing robustness), BREW prioritizes the integrity of the detection signal. By combining (1) a \textbf{distributed codeword architecture} that localizes errors and (2) a \textbf{window-shifting verification} mechanism that handles synchronization drift, BREW bridges the gap between zero-bit robustness and multi-bit utility. Our approach ensures that detection (confirming presence) logically precedes identification (decoding payload), thereby maintaining near-zero FPR without the catastrophic recall loss seen in thresholded extraction methods. 
This design intentionally favors reliable verification over maximizing raw payload capacity. Designated verification restricts acceptance to a single keyed target codeword per block, rather than accepting any valid codeword within the decoder radius. This reduces usable payload capacity compared to high-capacity ECC extractors, but it is precisely this restricted acceptance space that enables strong false-positive control. Longer payloads can still be distributed across multiple blocks, yielding a flexible capacity--reliability trade-off.


\section{Proposed Watermarking Framework}
\label{sec:method}

\begin{figure*}[t]
    \centering
    \vspace{0.1in}
    \includegraphics[width=\textwidth]{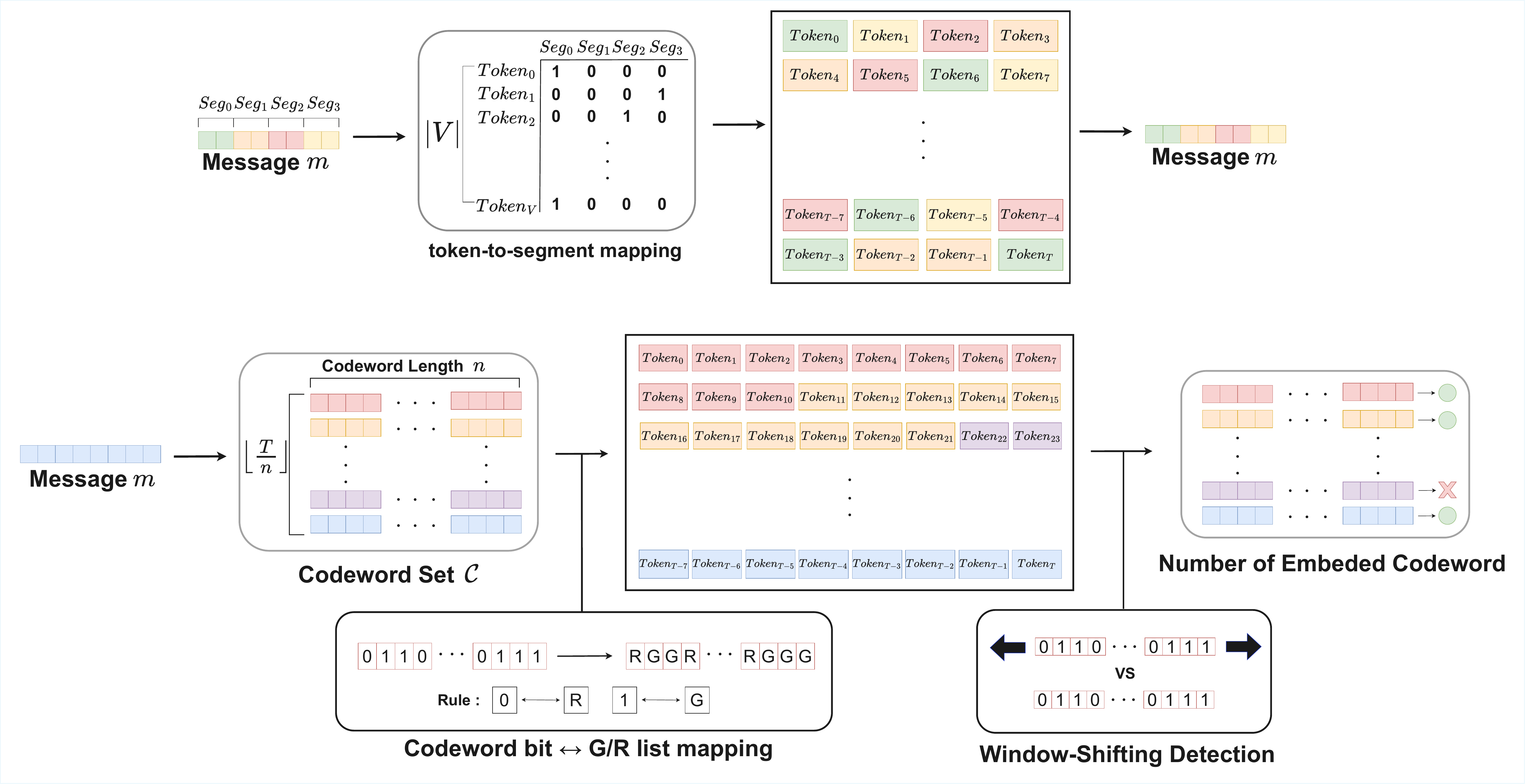}
    \caption{Comparison of multi-bit watermarking frameworks.
\textbf{(Top)} Prior schemes map every token to a segment, allowing ECC to ``correct'' accumulated noise into valid codewords, leading to high false positives.
\textbf{(Bottom)} BREW employs distributed embedding and window-shifting with designated verification. This eliminates ``any-codeword'' acceptance, preserving payload capacity while strictly controlling false positives.}
    \label{fig:overview}
    \vspace{0.1in}
\end{figure*}

We propose \textbf{BREW}, a \emph{reliable multi-bit} watermarking framework that explicitly targets the high FPR pitfall observed in prior multi-bit schemes, while preserving high TPR and robustness to common edits. The method has three pillars: 
(i) \emph{distributed resilience} via independent codeword blocks that enable partial watermark recovery and progressive confidence assessment even when some blocks are corrupted,
(ii) a \emph{window-shifting} detector that realigns and recovers individual codewords after insertion/deletion, contributing to the overall watermark strength score, and
(iii) a \emph{graduated verification protocol}, which quantifies watermark evidence by counting correctly matched \emph{designated} codewords rather than accepting ``any'' decodable codeword, thereby enabling continuous watermark strength measurement while suppressing spurious detections.
This incremental approach transforms binary detection into progressive evidence accumulation, where each recovered block contributes to a quantifiable confidence score. This section provides detailed algorithmic descriptions and technical analysis of each component.

\subsection{Reliable Multi-bit Detection via Designated-Codeword Verification}
\label{sec:method:dci}

Prior methods rely on previous tokens to collect information for codeword decoding, treating any text—watermarked or not—identically: the same token contributes to the decoding process, and ECC even corrects ``errors'' to produce false positives by reconstructing valid codewords from random noise. In contrast, our approach considers not only tokens but also their relative positions to verify whether patterns match the actual codeword structure, accepting only the designated codeword as a true detection rather than any valid codeword, thereby significantly reducing false positives while maintaining multi-bit capacity.

\subsubsection{Two-Stage Blind Detection Procedure}
\label{sec:method:two_stage}

We retain the use of meaningful messages $m\in\{0,1\}^k$ but depart from prior approaches by introducing a \textbf{two-stage detection protocol} that eliminates the need for an oracle (i.e., prior knowledge of the embedded message).

\textbf{Stage 1: Blind Message Estimation.} 
Since the watermark is embedded continuously across $M$ blocks, we first treat each block as an independent voter. 
The detector performs standard decoding on all blocks to collect candidate messages. 
The most frequent candidate $\hat{m}$ is selected via majority voting:
\begin{equation}
    \hat{m} = \arg\max_{m' \in \{0,1\}^k} \sum_{j=1}^{M} \mathbb{I}(\text{Decode}(b^{(j)}) = m').
\end{equation}
This step allows the detector to blindly recover the likely payload from noisy text.


\textbf{Stage 2: Designated Verification (The BREW Core).} 
Using the estimated payload $\hat{m}$, the detector reconstructs the 
\emph{designated codeword sequence} $\mathcal{Q}$ specific to the text's block indices using the shared key $\mathcal{K}$. 
For each block $j$, the detector deterministically derives a block-specific mask 
$r^{(j)} \in \{0,1\}^k$ from $\mathcal{K}$ and $j$, and uses it to reconstruct the designated codeword for that block.
We then apply our window-shifting verification (Algorithm~\ref{alg:detection}) to strictly validate $\hat{m}$. 
Specifically:
\begin{enumerate}
    \item \textbf{Reconstruction:} Generate the designated codeword 
    $c^{(j)} \leftarrow E(\hat{m} \oplus r^{(j)})$ for each block $j$.
    \item \textbf{Shift-aware Verification:} For each shared candidate offset 
    $s\in\mathcal{S}=[-s_{\max},s_{\max}]$, check whether the window centered at 
    $j\cdot n+s$ exactly matches $c^{(j)}$.
    \item \textbf{Decision:} The text is deemed watermarked only if the ratio of matched blocks exceeds a threshold $\theta$.
\end{enumerate}
This two-stage approach ensures that even if $\hat{m}$ is incorrect, e.g., in unwatermarked text, Stage 2 fails to find matching designated codewords, thereby keeping the FPR near zero. Importantly, BREW performs multi-bit payload recovery rather than only zero-bit detection: successful detection requires consistency between the recovered payload and the embedded blockwise codewords. We report message-level exact match rates in Appendix~\ref{app:match_rate}, showing that detection remains closely aligned with payload recovery even under paraphrasing attacks.

\subsection{Distributed Codeword Embedding}

\subsubsection{Vocabulary Partitioning Strategy}

Following the established approach of KGW \cite{kirchenbauer2023watermark}, we partition the vocabulary $\mathcal{V}$ into two disjoint sets for each block. For block $j$, we compute a block-specific seed such that $\text{seed}_j = H(\mathcal{K}, j)$, where $H$ is a cryptographic hash function (e.g., SHA-3) and $\mathcal{K}$ is the secret watermarking key. Using $\text{seed}_j$, we deterministically partition the vocabulary as 
$\mathcal{L}_0^{(j)} = \{v \in \mathcal{V} : H(\text{seed}_j, v) \bmod 2 = 0\}$ and
$\mathcal{L}_1^{(j)} = \{v \in \mathcal{V} : H(\text{seed}_j, v) \bmod 2 = 1\}$. This block-specific partitioning prevents adversaries from inferring vocabulary assignments across multiple generations, even with partial knowledge of the partitioning strategy.

\subsubsection{Codeword Generation and Selection}
We pre-compute a set of diverse codewords to avoid statistical patterns that could be exploited by adversaries. 
Specifically, the codeword generation strategy serves two purposes: 
(1) excluding all-zero codewords prevents degenerate cases that could impact detection accuracy, and 
(2) generating codeword pairs with maximum Hamming distance enhances robustness by ensuring diverse bit patterns. 
The detailed generation procedure is provided in Appendix~\ref{Appendix:codeword_generation}.


\begin{algorithm}[tb]
\caption{Distributed Watermark Embedding}
\label{alg:embedding}
\begin{algorithmic}[1]
\STATE \textbf{Input:} key $\mathcal{K}$, queue $\mathcal{Q}$, code length $n$, generation length $T$, strength $\delta$, scheme $\in\{\text{soft},\text{hard}\}$
\STATE \textbf{Output:} sequence of watermarked tokens $y_{1:T}$
\FOR{$t=1,2,\dots,T$}
  \STATE Obtain logits $\ell^{(t)}$ from the LM given current context
  \STATE $(j,b) \leftarrow \big(\lfloor (t-1)/n \rfloor,\ (t-1)\bmod n\big)$
  \STATE \textbf{If needed, initialize block $j$:} sample $c\in\mathcal{C}$, set $\mathcal{Q}[j]\!\leftarrow\!c$, and build $(\mathcal{L}^{(j)}_0,\mathcal{L}^{(j)}_1)$ using seed $H(\mathcal{K},j)$
  \STATE $z \leftarrow \mathcal{Q}[j][b]$ \COMMENT{target bit}
  \STATE Bias logits toward $\mathcal{L}^{(j)}_z$ with strength $\delta$ (and \textbf{mask} to $\mathcal{L}^{(j)}_z$ if scheme = hard)
  \STATE Sample token $y_t$ from the resulting distribution and append it to the sequence
\ENDFOR
\end{algorithmic}
\end{algorithm}

\subsubsection{Distributed Embedding Algorithm}
Our embedding algorithm generates text in blocks of length $n$, with each block embedding exactly one codeword (Algorithm~\ref{alg:embedding}). 
We support two embedding schemes. \textbf{Soft embedding} adds a bias $\delta$ to the logits of the target list before softmax, preserving natural variation while encouraging codeword-consistent token selection. In contrast, \textbf{hard embedding} restricts sampling entirely to the target list, ensuring exact codeword embedding at the potential cost of reduced text quality. 
This design enables a tunable trade-off between watermark strength and linguistic naturalness.

\begin{algorithm}[tb]
  \caption{Synchronized Sliding Window Detection}
  \label{alg:detection}
  \begin{algorithmic}[1]
    \STATE \textbf{Input:} tokens, key $\mathcal{K}$, code parameters $(n,k)$, shift budget $s_{\max}$, threshold $\theta$
    \STATE \textbf{Output:} decision \texttt{is\_wm} and estimated payload $\hat{m}_{final}$
    \STATE $\mathcal{S}\leftarrow\{s\in\mathbb{Z}:|s|\le s_{\max}\}$, $best\_ratio \leftarrow 0$, $\hat{m}_{final} \leftarrow \textbf{None}$
    \FOR{$s \in \mathcal{S}$} 
        \STATE $\mathbf{b}_s \leftarrow \textsc{ExtractBitStream}(\text{tokens}, \mathcal{K}, \text{offset}=s)$, $M \leftarrow \lfloor |\mathbf{b}_s|/n \rfloor$, $V[\cdot]\leftarrow 0$
        \FOR{$j=0,\dots,M-1$}
          \STATE $m' \leftarrow \textsc{DecodeAndExtract}(\mathbf{b}_s[j n:(j{+}1)n])$
          \STATE \textbf{if} $m' \neq \textbf{None}$ \textbf{then} $V[m'] \leftarrow V[m'] + 1$
        \ENDFOR
        \STATE $\hat{m}_s \leftarrow \arg\max_{m'} V[m']$, Reconstruct $\mathcal{Q}$ using $\hat{m}_s$ and offset $s$, $matched \leftarrow 0$
        \FOR{$j=0,\dots,M-1$}
          \STATE \textbf{if} $\textsc{SafeDecode}(\mathbf{b}_s[jn : jn+n]) = \mathcal{Q}[j]$ \textbf{then} $matched \leftarrow matched + 1$
        \ENDFOR
        \STATE \textbf{if} $matched/M \ge best\_ratio$ \textbf{then} $best\_ratio \leftarrow matched/M$, $\hat{m}_{final} \leftarrow \hat{m}_s$
    \ENDFOR
    \STATE $\texttt{is\_wm} \leftarrow (best\_ratio \ge \theta)$
    \STATE \textbf{if} $\neg \texttt{is\_wm}$ \textbf{then} $\hat{m}_{final} \leftarrow \textbf{None}$
    \STATE \textbf{Return} $\texttt{is\_wm}, \hat{m}_{final}$
  \end{algorithmic}
\end{algorithm}

\subsection{Sliding Window Synchronization for Linear Shift Recovery}
\label{sec:method:window}

Our \textbf{Sliding Window Synchronization} mechanism ensures robustness against insertion and deletion attacks. Unlike prior circular-shift assumptions, edits induce \emph{linear shifts} across block boundaries, rendering static slicing ineffective. BREW addresses this by treating the bit sequence as a continuous stream. 
For detection, BREW uses a single global candidate offset set $\mathcal{S}=\{-s_{\max},\ldots,s_{\max}\}$. For each global offset candidate $s\in\mathcal{S}$, the detector applies the same offset $s$ to the entire token sequence to extract a shifted bit stream $\mathbf{b}^{(s)}$. It then partitions $\mathbf{b}^{(s)}$ into consecutive length-$n$ blocks $\mathbf{b}^{(s)}[jn:(j+1)n]$ for all $j$. Thus, $s$ is shared across all blocks and is not selected independently for each block.
For analysis, we upper-bound this global-offset detector by a more permissive per-block search model, where each block is allowed to choose its best offset in $\mathcal{S}$ independently. This relaxation can only increase the number of accidental matches, since maximizing per block is at least as permissive as choosing one global offset for all blocks. Therefore, the FPR bound derived under the per-block model is a conservative upper bound for the actual global-offset implementation.




\paragraph{Safe Decoding with Error Handling.}
To ensure robust \emph{incremental detection} even when individual blocks contain uncorrectable errors, we implement a safe decoding subroutine that gracefully handles decoder failures. The decoder accepts only codewords within the correction radius $t$ and returns \texttt{None} otherwise, preventing spurious matches while allowing other blocks to contribute to the watermark strength score. The concrete decoding procedure and full algorithm are provided in Appendix~\ref{Appendix:safe_decode}.

\paragraph{Mechanism of Linear Window Search.}
Our core detection algorithm~\ref{alg:detection} augments standard error-correcting decoding with systematic linear windowing.

\textit{Linear Search Rationale}: When insertions or deletions introduce a net shift in the extracted bit stream, the correct block boundaries may be displaced from their nominal positions. By evaluating global offset candidates $s\in[-s_{\max},s_{\max}]$, the detector tests shifted framings of the entire bit stream and selects the framing with the highest aggregate designated-codeword match ratio. This recovers synchronization when the net shift is within the search budget.


\paragraph{Limitations and Adaptive Extension.}
The current fixed-range search is designed for bounded global desynchronization. However, in extremely long texts where accumulated insertions or deletions exceed $s_{\max}$, the correct global framing may fall outside the candidate range. A possible future extension is to use adaptive synchronization that updates the global search center across longer spans of text. This extension is orthogonal to our current detector, which evaluates a shared global offset candidate over all blocks.

\subsection{Parameter Selection and Optimization}
We provide general guidelines for parameter selection, focusing on block length $n$, 
bias parameter $\delta$, and maximum shift $s_{\max}$. 
These parameters govern the trade-off between robustness, detection accuracy, and text quality. 
Comprehensive trade-off analyses and recommended configurations are deferred to Appendix~\ref{Appendix:parameter}.

\subsection{Computational Complexity Analysis and Security Properties}
The embedding procedure has the same $O(|\mathcal{V}|)$ per-token complexity as existing methods, 
while detection introduces an additional factor proportional to the maximum shift $s_{\max}$. 
Formal derivations and detailed complexity expressions are given in Appendix~\ref{Appendix:complexity}.
Our approach inherits the security guarantees of the underlying hash function and error-correcting code, while introducing additional resilience via block-wise embedding and codeword diversity. 
A full discussion of key security, codeword diversity, and block independence is provided in Appendix~\ref{Appendix:security}.

\section{Analytical Bounds for FPR/FNR in ECC-Backed Watermarks}
\label{sec:theory}

This section develops finite-sample bounds for the proposed watermarking scheme based on blockwise \emph{codeword-presence} detection with window-shifting. We quantify false-positive (FPR) and false-negative (FNR) probabilities under general $q$-ary linear codes, and isolate the role of the embedding bias parameter $\delta$ in the soft-embedding regime. We summarize here the setup and key intuition, while deferring detailed theorems and proofs to Appendix~\ref{Appendix:finite}.

\paragraph{Setup and Notation.}
We consider a $q$-ary linear block code $C\subseteq \Sigma^n$ with unique-decoding radius $t$. Each text block embeds a designated codeword via $\delta$-biased sampling from a green/red partition of the vocabulary. Detection is performed by unique decoding with window-shifting to counter misalignments. (See Appendix~\ref{Appendix:finite:setup} for detailed definitions.)

\paragraph{False Positives.}
We analyze two types of tests: (i) a naïve ``any-codeword'' presence test, and (ii) the proposed designated-codeword test with window-shifting.  
Theorems~\ref{thm:any} and \ref{thm:single_fpr} (Appendix~\ref{Appendix:finite:any}, \ref{Appendix:finite:designated}) quantify single-block and aggregate FPR under these schemes, highlighting exponential suppression in the block length $n$ and the number of blocks $M$. 
Importantly, the aggregate independence assumption is imposed on the detector-facing matching indicators induced by block-specific hash partitions, not on the raw natural-language tokens; Appendix~\ref{Appendix:finite:aggfpr} provides the detailed justification.

\paragraph{False Negatives.}
The impact of soft embedding ($\delta$-bias) and adversarial edits is modeled via an effective symbol error probability $p_{\mathrm{tot}}$.  
Theorem~\ref{thm:fnr} (Appendix~\ref{Appendix:finite:fnr}) shows that the aggregate FNR decays exponentially in $M$ provided that $p_{\mathrm{tot}}<t/n$.

\paragraph{Design Implications.}
The combined FPR/FNR bounds yield a clear design rule: choose parameters $(n,t,s_{\max},\theta,M,\delta)$ so that $\theta$ balances the two Chernoff exponents, and $\delta$ is large enough to keep the embedding error below $t/n$. See Appendix~\ref{Appendix:finite:params} for proofs, examples, and entropy-based parameter guidelines.

\section{Results}
\label{sec:experiments}

\subsection{Experimental Setup}

\paragraph{Models and Datasets.} 
We evaluate on OPT-1.3B~\cite{zhang2022opt}, LLaMA-3.2-3B~\cite{touvron2023llama}, and Mistral-7B~\cite{jiang2023mistral7b} using the C4 and OpenGen datasets. Unless otherwise specified, OPT-1.3B on C4 is used as the default setting. \textit{Scalability and Model-Agnosticism:} BREW is \emph{model-agnostic}, operating exclusively on output logits independent of model architecture or size. Consequently, the observed performance trade-offs theoretically transfer to larger foundation models (e.g., 70B+), where higher entropy further facilitates embedding.

\paragraph{Baselines.}
We compare BREW against representative multi-bit watermarking methods.
Specifically, we include \textbf{MPAC}~\cite{yoo2024advancing}, a non-ECC multi-bit watermarking approach based on invariant linguistic features, and \textbf{\cite{qu2025provably}}, a recent ECC-based method using Reed--Solomon codes. For standard comparisons, we follow the official implementations with identical parameter settings. Crucially, to address potential concerns that the high FPR of baselines stems solely from the lack of a rejection threshold, we rigorously evaluate whether their performance can be redeemed by simple thresholding (a ``steel-manning'' stress test). We analyze the full TPR--FPR trade-off via ROC curves to verify if an optimal threshold exists or if the method fundamentally fails to separate watermarked from unwatermarked text under attacks.
\footnote{We also considered RBC~\cite{chao2024rbc}, which employs LDPC codes; however, we exclude it from evaluation because no public implementation is available and the reported LDPC parameters ($n=12$, $k=5$, $d_v$=3, $d_c$=4) violate the standard constraint $n \cdot d_v = (n - k) \cdot d_c$ (36 $\neq$ 28).}

\paragraph{Parameters and Detection.}
We adopt BCH($n=31$, $k=6$, $t=7$) for embedding due to its favorable sensitivity--FPR trade-off. We use \emph{soft} watermarking (defaulting over hard for better text quality) with insertion strength $\delta\in\{1.5,2.0,3.0,6.0\}$, with sensitivity to $\delta$ analyzed in Appendix~\ref{Appendix:Window_shift_Range_delta}.

During detection, we use a window-shift range of $s_{\max}\in\{0,1,3,5\}$ to recover token–codeword alignment under insertion or deletion attacks (Appendix~\ref{Appendix:Window_shift_Range}). 
Detection follows an \emph{incremental} protocol that counts recovered codewords: by default, a text is declared watermarked if at least one codeword is recovered, while a stricter two-codeword threshold is analyzed in Appendix~\ref{Appendix:threshold_two_codewords}. 
We adopt the structured detector in all experiments, as the na\"{i}ve variant yields consistently high false positive rates (Appendix~\ref{Appendix:no_attack}). 
We evaluate texts truncated to $T\in\{200,500\}$ and report standard detection metrics (TPR, FPR, Precision, and F1), with primary analysis based on ROC curves.

\paragraph{Implementation Details.}

All experiments are conducted using the MarkLLM evaluation framework~\cite{pan-etal-2024-markllm}, 
with all baseline methods evaluated via their official implementations and our watermark embedding and detection algorithms reimplemented as custom modules under the same evaluation pipeline.

\subsection{Effect of Codeword Parameters}
\label{sec:codeword_effect}

\begin{table}[t]
  \caption{Detection performance (TPR/FPR) of different BCH codeword configurations at $\delta=3$ and $s_{\max}=5$ under token-increasing (insertion-like) synonym substitution attacks ($T=200$).}
  \label{table:bch_delta3_s5}
  \centering
  \resizebox{\columnwidth}{!}{%
  \begin{tabular}{l ccc ccc}
    \toprule
    & \multicolumn{3}{c}{\textbf{5\% Insertion}}
    & \multicolumn{3}{c}{\textbf{10\% Insertion}} \\
    \cmidrule(lr){2-4}\cmidrule(lr){5-7}
    & \textbf{(15,5,3)} & \textbf{(31,6,7)} & \textbf{(63,7,15)}
    & \textbf{(15,5,3)} & \textbf{(31,6,7)} & \textbf{(63,7,15)} \\
    \midrule
    \textbf{TPR}
    & 0.990 & 0.930 & 0.710
    & 1.000 & 0.710 & 0.305 \\
    \textbf{FPR}
    & 0.945 & 0.085 & 0.000
    & 0.930 & 0.110 & 0.005 \\
    \bottomrule
  \end{tabular}}
\end{table}
Before evaluating robustness under various attacks, we examine the effect of BCH codeword parameters on detection performance to identify a suitable default configuration under token-altering perturbations. We compare three BCH codeword configurations at $T=200$: a short $(15,5,3)$, a medium $(31,6,7)$, and a long $(63,7,15)$ codeword. As shown in Table~\ref{table:bch_delta3_s5}, shorter codewords tend to overreact to noise, resulting in excessively high false positive rates, whereas longer codewords overly suppress false positives at the cost of degraded true positive rates under insertion attacks. The medium configuration strikes a favorable balance between these extremes, maintaining high detection sensitivity while keeping false positives at a moderate level. Accordingly, we adopt the $(31,6,7)$ codeword as the default setting in subsequent experiments; full results including deletion-like substitutions are provided in Appendix~\ref{Appendix:codeword_para}.

\subsection{Synonym Substitution Attack}

\begin{figure}[t]
    \centering
    \includegraphics[width=\columnwidth]{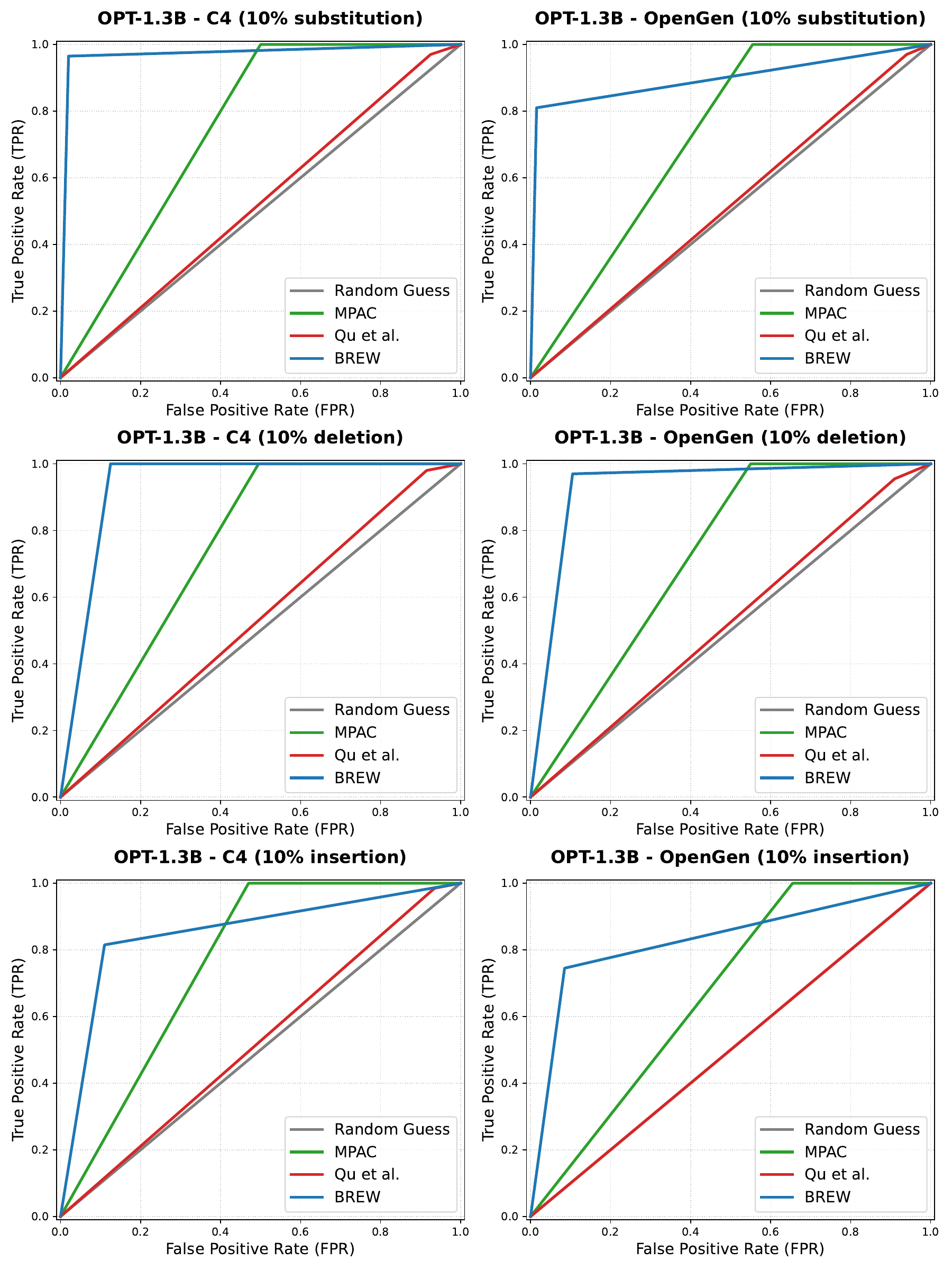}
    \caption{
    ROC curves under 10\% synonym substitution attacks on the OPT-1.3B model with text length $T=200$.
    \textbf{Top:} token-preserving substitutions;
    \textbf{Middle:} token-reducing (deletion-like) substitutions;
    \textbf{Bottom:} token-increasing (insertion-like) substitutions.
    Columns correspond to the C4 (left) and OpenGen (right) datasets.
    The figure compares detection performance of BREW, MPAC~\cite{yoo2024advancing}, and \cite{qu2025provably}.}
    \label{fig:ROC_OPT_1_3B_Final}
\end{figure}

We evaluate robustness under synonym substitution attacks, in which words in a watermarked text are replaced with semantically equivalent alternatives. Although substitutions are applied at the word level, they induce three distinct effects at the token level: (i) token-preserving, (ii) token-reducing (deletion-like), and (iii) token-increasing (insertion-like), which either preserve or disrupt token--codeword alignment. Experiments are conducted on the C4 and OpenGen datasets with substitution rates of 5\% and 10\%. Unless otherwise stated, MPAC~\cite{yoo2024advancing} and \cite{qu2025provably} are evaluated using their recommended watermark strength $\delta=6$, while BREW is evaluated with $\delta=6$ and a window-shift parameter $s_{\max}=5$ to tolerate limited token insertion and deletion effects. In the main text, we report results for the OPT-1.3B model under the 10\% substitution setting with text length $T=200$, which represents a challenging and practically relevant attack regime; comprehensive results covering additional substitution rates (5\%), text lengths ($T=500$), and backbone models are deferred to Appendix~\ref{Appendix:synsub_figures}. Detection performance is analyzed using ROC curves to characterize the trade-off between true positive and false positive rates.

\subsubsection{Token-Preserving Synonym Substitution}

Token-preserving substitutions maintain token--codeword alignment but perturb the statistical distribution of generated tokens. As illustrated in the \textbf{top} of Figure~\ref{fig:ROC_OPT_1_3B_Final}, BREW achieves a clear separation between watermarked and unwatermarked texts under 10\% synonym substitution on OPT-1.3B with $T=200$. MPAC~\cite{yoo2024advancing} exhibits intermediate performance, improving over \cite{qu2025provably} but remaining less reliable than BREW. In contrast, \cite{qu2025provably} behaves close to random-guessing, reflecting extremely high false positive rates even in the absence of token-level desynchronization. This result indicates that distributional perturbations alone are sufficient to induce frequent false detections in existing multi-bit schemes, whereas BREW remains robust when token alignment is preserved.

\subsubsection{Token-Altering Synonym Substitution (Deletion/Insertion-like)}

Token-altering substitutions disrupt alignment, making detection challenging. However, as shown in Figure~\ref{fig:ROC_OPT_1_3B_Final} (\textbf{middle}, \textbf{bottom}), BREW remains robust under deletion- and insertion-like attacks, maintaining favorable ROC separation via window-shift realignment (Appendix~\ref{Appendix:Window_shift_Range_delta}). MPAC~\cite{yoo2024advancing} shows intermediate performance, partially mitigating disruption but proving less reliable than BREW. In contrast, \cite{qu2025provably} approaches random guessing, confirming that token-level desynchronization leads to pervasive false detections.

\subsubsection{Analysis of Thresholding Effects (Steel-Manning)}
\label{sec:steelmanning}

A critical question regarding baseline fairness is whether the high FPR of \cite{qu2025provably} can be mitigated simply by applying a rejection threshold. Our ROC analysis in Figures~\ref{fig:ROC_OPT_1_3B_Final} and \ref{fig:Insertion_ROC_Final} serves as a ``steel-manning'' evaluation, sweeping all possible thresholds.
The results reveal that for \cite{qu2025provably}, TPR and FPR increase linearly together (following the random-guess diagonal $y=x$), implying that no optimal threshold exists to trade off FPR for TPR; strictly lowering FPR causes a proportional, catastrophic drop in TPR. In contrast, BREW's curve bows significantly toward the upper-left corner, maintaining high TPR ($>0.96$) even at strict low-FPR thresholds ($0.02$). This confirms that BREW's advantage is structural—stemming from \emph{window-shifting alignment}—and cannot be replicated in prior methods merely by tuning detection thresholds.

\subsection{Paraphrasing Attack}

\begin{figure}[t]
    \centering
    \includegraphics[width=\columnwidth]{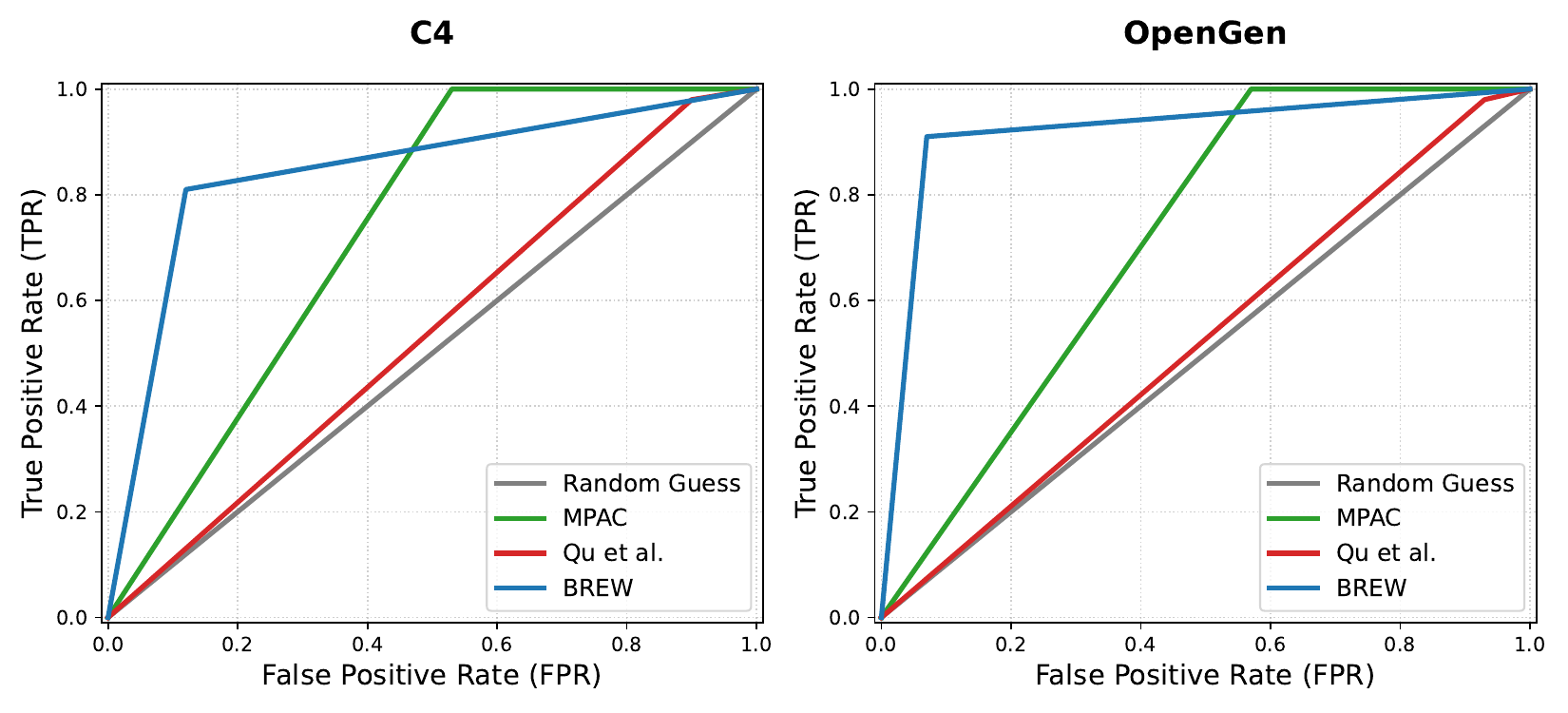}
    \caption{
    ROC curves under paraphrasing attacks on the OPT-1.3B model evaluated on the C4 and OpenGen datasets.
    The figure compares BREW, MPAC~\cite{yoo2024advancing}, and \cite{qu2025provably}.
    Detailed numerical results are provided in Appendix~\ref{Appendix:Paraphrasing_attacks}.
    }
    \label{fig:Paraphrasing_Attack}
\end{figure}

Paraphrasing attacks challenge detection by simultaneously altering token identity, length, and ordering. We evaluate robustness using a T5-based paraphraser~\cite{raffel2020exploring} on OPT-1.3B outputs from C4 and OpenGen. Figure~\ref{fig:Paraphrasing_Attack} shows that BREW maintains strong performance in the low false positive rate regime across both datasets. While MPAC~\cite{yoo2024advancing} attains higher TPR, it suffers from substantially increased false positives. In contrast, \cite{qu2025provably} approaches random guessing due to pervasive false detections. Overall, BREW proves robust to semantic rewriting by explicitly controlling false positives through structured evidence accumulation.

\section{Conclusion}
\label{sec:conclusion}

We proposed an incremental detection framework that addressing the high FPR of prior ECC-based watermarking. By shifting to designated verification, BREW effectively suppresses false positives while maintaining robust detection under attacks. Experiments confirm near-zero FPR and high TPR, significantly outperforming baselines such as MPAC and \cite{qu2025provably}. Crucially, our \emph{model-agnostic} design ensures scalability to large foundation models. While BREW prioritizes reliable short-payload verification under strict FPR control over very high payload capacity, longer payloads can be distributed across multiple blocks, yielding a flexible capacity--reliability trade-off. Future work will explore adaptive windowing to mitigate accumulated global drift in long texts while balancing alignment robustness and spurious-match risk.


\section*{Impact Statement}
This paper presents a method for improving the reliability of multi-bit watermark detection in large language models.
By focusing on false positive control under strong text transformations, such as paraphrasing and token-altering attacks, the proposed approach aims to support responsible provenance identification of machine-generated text.
In particular, reducing false positives helps minimize the risk of incorrectly attributing AI-generated content to human authors, which is an important consideration in academic and professional settings.
We do not foresee significant negative societal impacts beyond those already associated with text watermarking technologies, and the method is intended to complement existing efforts toward transparency and accountability in AI-generated content.

\section*{Acknowledgments}
This work was partly supported by the Institute of Information \& communications Technology Planning \& Evaluation (IITP) grant funded by the Korea government (MSIT) (RS-2024-00399401, Development of Quantum-Safe Infrastructure Migration and Quantum Security Verification Technologies, 40\%), the Institute of Information \& communications Technology Planning \& Evaluation (IITP) grant funded by the Korea government (MSIT) (RS-2024-00442085, Development of V2X Infra Security Core Technologies for Autonomous Vehicle Services, 40\%), and the Regional Innovation System \& Education (RISE) program through the Ulsan RISE Center, funded by the Ministry of Education (MOE) and the Ulsan Metropolitan City, Republic of Korea (2026-RISE-07-001, 20\%).




\nocite{langley00}

\bibliography{ICML2026_BCH_paper_ref}
\bibliographystyle{icml2026}

\newpage
\appendix
\onecolumn



\section{Detailed Algorithms}
\label{Appendix:algorithms}




\subsection{Diverse Codeword Generation}
\label{Appendix:codeword_generation}

This algorithm provides the detailed procedure for generating diverse codewords. As described in Section~\ref{sec:method}, we exclude the all-zero codeword and use maximum-weight pairs to maximize Hamming distance and improve robustness.

The operation $c_2 = c_1 \oplus c_{\max}$ may produce the all-zero codeword when $c_1 = c_{\max}$. To prevent this degenerate case, Algorithm~\ref{alg:codeword_generation} explicitly guards against $c_2=\mathbf{0}$ and selects the nonzero codeword $c_1$ instead. This ensures that the generated codeword set never contains the all-zero codeword, maintaining consistency with the designated-codeword detector.

\begin{algorithm}[tb]
  \caption{Diverse Codeword Generation}
  \label{alg:codeword_generation}
  \begin{algorithmic}
    \STATE {\bfseries Input:} error-correcting code $\mathcal{C}$ with parameters $(n,k,t)$, secret key $\mathcal{K}$
    \STATE {\bfseries Output:} diverse codeword set $\mathcal{Q}$

    \STATE Define message space $\mathcal{M} \leftarrow \{0,1\}^k \setminus \{0^k\}$
    \STATE Initialize codeword set $\mathcal{Q} \leftarrow \emptyset$
    \STATE Find maximum-weight codeword
    \STATE \hspace{1em}$c_{\max} \leftarrow \arg\max_{c \in \mathcal{C}} \mathrm{wt}(c)$

    \WHILE{$|\mathcal{Q}| <$ required number of blocks}
      \STATE Sample a random message $m$ uniformly from $\mathcal{M}$
      \STATE Encode the message to obtain $c_1 \leftarrow \mathcal{E}(m)$
      \STATE Compute a paired codeword $c_2 \leftarrow c_1 \oplus c_{\max}$
      \IF{$c_2 = \mathbf{0}$}
        \STATE Set $c \leftarrow c_1$ \COMMENT{avoid the all-zero codeword}
      \ELSE
        \STATE Randomly select $c \in \{c_1, c_2\}$
      \ENDIF
      \STATE $\mathcal{Q} \leftarrow \mathcal{Q} \cup \{c\}$
    \ENDWHILE
  \end{algorithmic}
\end{algorithm}

\subsection{Bit Sequence Extraction}
\label{Appendix:extraction}

For completeness, we provide the full pseudocode of the bit extraction procedure that maps generated tokens to binary sequences and segments them into fixed-length blocks. This expands on the conceptual description given in Section~\ref{sec:method}.

\begin{algorithm}[tb]
  \caption{Continuous Bit Stream Extraction with Alignment}
  \label{alg:extraction}
  \begin{algorithmic}
    \STATE {\bfseries Input:} text tokens $\{s_0,\ldots,s_T\}$, secret key $\mathcal{K}$, alignment offset $\Delta$
    \STATE {\bfseries Output:} continuous bit stream $\mathbf{b} \in \{0,1\}^U$

    \STATE $N_p \leftarrow$ prompt length
    \STATE $U \leftarrow T - N_p$
    \STATE Initialize bit array $\mathbf{b}$ of size $U$
    \STATE $j_{\text{prev}} \leftarrow -1$

    \FOR{$t = N_p$ {\bfseries to} $T-1$}
      \STATE $idx \leftarrow t - N_p$
      \STATE \COMMENT{Adjust block index based on the alignment offset $\Delta$}
      \STATE $j \leftarrow \lfloor (idx + \Delta) / n \rfloor$ 
      
      \IF{$j \neq j_{\text{prev}}$}
        \STATE $\text{seed}_j \leftarrow H(\mathcal{K}, j)$
        \STATE Build partitions $\mathcal{L}_0^{(j)}, \mathcal{L}_1^{(j)}$ using $\text{seed}_j$
        \STATE Precompute lookup function $f_j(v) \in \{0,1\}$
        \STATE $j_{\text{prev}} \leftarrow j$
      \ENDIF
      \STATE $\mathbf{b}[idx] \leftarrow f_j(s_t)$
    \ENDFOR

    \STATE \textbf{return} $\mathbf{b}$
  \end{algorithmic}
\end{algorithm}

\subsection{Safe Error-Correcting Decoder}
\label{Appendix:safe_decode}

This algorithm expands on the safe decoding strategy summarized in Section~\ref{sec:method}. 
It ensures robust handling of uncorrectable blocks by returning \texttt{None} when decoding exceeds the correction radius.

\begin{algorithm}[tb]
  \caption{Safe Error-Correcting Decoder}
  \label{alg:safe_decode}
  \begin{algorithmic}
    \STATE {\bfseries Input:} bit sequence $x \in \{0,1\}^n$, code $\mathcal{C}$, correction radius $t$
    \STATE {\bfseries Output:} $(\hat{c}, d)$ if decodable within $t$; otherwise \texttt{None}

    \STATE Ensure field and parity-check structures for $\mathcal{C}$ are initialized
    \STATE Attempt to decode $x$ using $\mathcal{C}$ and obtain $(\hat{c}, d)$

    \IF{decoding fails}
      \STATE \textbf{return} \texttt{None}
    \ENDIF

    \IF{$d \le t$}
      \STATE \textbf{return} $(\hat{c}, d)$
    \ELSE
      \STATE  \textbf{return} \texttt{None}
    \ENDIF
  \end{algorithmic}
\end{algorithm}

\section{Additional Analyses}
\label{Appendix:additional}

\subsection{Error-Correcting Code Selection and Parameterization}
\label{Appendix:ecc}

We primarily employ BCH codes \cite{blahut1983theory} due to their well-understood properties and efficient implementation. 
For code length $n = 2^m - 1$, we select parameters based on the trade-off between error-correction capability and false-positive rates:

\begin{itemize}
\item \textbf{BCH(15,5,3)}: Short blocks with moderate error-correction capability, suitable for shorter text segments.
\item \textbf{BCH(31,6,7)}: The default configuration in our experiments, balancing block length, robustness, and detection reliability.
\item \textbf{BCH(63,7,15)}: A longer-block configuration with higher error-correction capability, evaluated for stronger attack scenarios.
\end{itemize}


\subsubsection{Alternative Error-Correcting Codes}

Our framework readily accommodates other linear block codes \cite{Richardson2008ecc}:

\textbf{Reed--Solomon Codes}: Optimal for burst error correction, particularly effective when insertion/deletion attacks create localized corruption patterns.

\textbf{LDPC Codes}: Superior performance for longer blocks, but increased computational complexity. Recommended for applications requiring very low false positive rates.

\textbf{Convolutional Codes}: Well-suited for streaming applications where text is generated and detected incrementally.

\textbf{Code Selection Guidelines}:
\begin{itemize}
\item Choose code length $n$ based on expected text length and block granularity requirements
\item Select error-correction capability $t$ based on anticipated attack strength
\item Balance code rate $k/n$ against false positive requirements using our theoretical analysis (Section~\ref{sec:theory})
\end{itemize}

\subsection{Parameter Selection and Optimization}
\label{Appendix:parameter}

\subsubsection{Block Length Optimization}

The choice of block length $n$ involves several trade-offs:

\textbf{Shorter blocks} ($n \leq 31$): 
\begin{itemize}
\item Advantages: Better localization of insertion/deletion effects, faster detection
\item Disadvantages: Higher false positive rates, reduced error-correction capability
\end{itemize}

\textbf{Longer blocks} ($n \geq 63$):
\begin{itemize}
\item Advantages: Lower false positive rates, stronger error correction
\item Disadvantages: Larger vulnerability to insertion/deletion within blocks
\end{itemize}

We recommend $n = 31$ for most applications, providing a good balance between robustness and efficiency.

\subsubsection{Bias Parameter Tuning}

The bias parameter $\delta$ controls the strength of watermark embedding:

\begin{itemize}
\item $\delta \in [1.5, 2.0]$: Minimal text quality impact, moderate watermark strength
\item $\delta \in [2.0, 2.5]$: Balanced trade-off for most applications  
\item $\delta \in [2.5, 3.0]$: Strong watermarking for high-security scenarios
\end{itemize}

\subsubsection{Window-Shift Range}
\label{Appendix:Window_shift_Range}

The maximum shift parameter $s_{\max}$ should be chosen based on expected insertion/deletion rates:

\begin{align}
s_{\max} &\geq \alpha \cdot n \cdot p_{\text{ins/del}}
\end{align}

where $p_{\text{ins/del}}$ is the expected insertion/deletion rate and $\alpha \geq 1.5$ provides a safety margin. 
In our experiments with $n=31$, we use $s_{\max}=5$ as a practical default, which provides a favorable balance between alignment recovery and spurious-match risk. Larger values can improve robustness to longer accumulated drift, but also increase the number of candidate framings and may raise FPR.

\subsection{Capacity--Reliability Trade-off}
\label{Appendix:capacity_tradeoff}

BREW intentionally prioritizes reliable short-payload verification under strict
false-positive constraints over maximizing raw payload capacity. Conventional
ECC-based extraction accepts any valid codeword within the correction radius,
whereas designated verification accepts only the keyed target codeword assigned
to each block. This reduces the acceptance space from roughly $2^k$ valid
codewords to one designated codeword at the block level.

In our default BCH$(31,6,7)$ setting, this corresponds to sacrificing up to
$k=6$ bits of unconstrained acceptance capacity per block, or approximately
$5.85$ effective bits after accounting for the nonzero-codeword construction
used in our implementation. This reduction is structural: it trades raw payload
capacity for strong false-positive control. For applications requiring longer
payloads, messages can be distributed across multiple independent blocks, while
the designated verification step maintains strict control over spurious
matches.

\subsection{Computational Complexity Analysis}
\label{Appendix:complexity}

\subsubsection{Embedding Complexity}

The computational overhead during text generation consists of:
\begin{itemize}
\item Hash computation: $O(1)$ per token
\item Vocabulary partitioning: $O(|\mathcal{V}|)$ per block, amortized $O(|\mathcal{V}|/n)$ per token
\item Logit modification: $O(|\mathcal{V}|)$ per token
\end{itemize}

Total embedding complexity: $O(|\mathcal{V}|)$ per token, the same as existing methods.

\subsubsection{Detection Complexity}

Detection complexity depends on the number of shift operations:
\begin{itemize}
\item Bit extraction: $O(T)$ for text length $T$
\item Error correction per block: $O(n^3)$ using standard algorithms
\item Window-shifting: $O(s_{\max} \cdot n^3)$ per block in the worst case
\end{itemize}

Total detection complexity: $O(T \cdot s_{\max} \cdot n^2)$, where the factor $s_{\max}$ represents the overhead of shift search. For practical parameters ($s_{\max} = 10$, $n = 31$), this remains computationally tractable.

\subsection{Security Properties}
\label{Appendix:security}
Our method inherits the cryptographic properties of the underlying hash function and error-correcting code while providing additional security benefits through block-wise design.

\textbf{Key Security}: The secret key $\mathcal{K}$ determines vocabulary partitioning and codeword selection. Without knowledge of $\mathcal{K}$, an adversary cannot distinguish watermarked from unwatermarked text beyond statistical artifacts under the assumed one-wayness of the cryptographic hash function.

\textbf{Codeword Diversity}: Random codeword generation prevents statistical attacks based on repeated patterns. Each text embeds different codewords, making it infeasible to infer watermarking parameters from multiple samples.

\textbf{Block Independence}: Unlike methods that embed a single codeword across multiple blocks, our approach ensures that the compromise of one block does not affect others, providing better security compartmentalization.

The following section provides a formal theoretical analysis of detection bounds and false positive rates under our framework.

\section{Finite-sample Bounds: Detailed Proofs and Examples}
\label{Appendix:finite}

This appendix contains the complete derivations, theorems, proofs, and examples for the finite-sample bounds introduced in Section~\ref{sec:theory}.

\subsection{Setup and Notation}
\label{Appendix:finite:setup}
Let $\Sigma=\{0,1,\dots,q-1\}$ and let $C\subseteq \Sigma^n$ be a $q$-ary linear block code with length $n$, dimension $k$, and minimum Hamming distance $d_{\min}$. Its unique-decoding radius is $t=\lfloor (d_{\min}-1)/2 \rfloor$. Define the $q$-ary Hamming ball volume as
\begin{equation}
V_q(n,t)\triangleq\sum_{i=0}^{t}\binom{n}{i}(q-1)^i.
\label{eq:Vq}
\end{equation}

A text is partitioned into $M$ disjoint blocks. For block $j\in\{1,\dots,M\}$, a secret seed $\mathrm{seed}_j$ (derived from a global key and block index) deterministically specifies (i) a single \emph{designated} codeword $c^{(j)}\in C$ to be embedded in that block and (ii) a partition of the vocabulary into green/red (or more generally $q$-ary) token lists aligned with the symbols of $c^{(j)}$. 

\paragraph{Embedding.}
In \emph{soft} embedding, logits of tokens in the green list are shifted by $+\delta$ while others are left unchanged, and a token is sampled from the resulting softmax. In \emph{hard} embedding, sampling is restricted to the green list (formally, $\delta\to\infty$).

\paragraph{Detection.}
Given a candidate text, the detector extracts a $q$-ary symbol sequence $b^{(j)}\in\Sigma^n$ from each block $j$ according to the green/red partition induced by $\mathrm{seed}_j$, and applies a unique decoder for $C$ to decide whether $b^{(j)}$ lies within Hamming distance $\le t$ from the designated codeword $c^{(j)}$. To counter local misalignments caused by insertions or deletions, the detector searches over a bounded set of linear offsets of magnitude at most $s_{\max}$ around the nominal block boundary. We denote by $S \triangleq 2s_{\max}+1$ the number of candidate offsets (including zero). The global decision is based on the fraction of blocks that decode successfully: if the match ratio exceeds a threshold $\theta\in(0,1)$, the text is declared watermarked. Although Algorithm~\ref{alg:detection} applies a single global offset across all blocks, the theoretical analysis considers a per-block shift search for tractability. Since taking the maximum over offsets per block yields a conservative upper bound, all FPR guarantees derived under the per-block model remain valid for Algorithm~\ref{alg:detection}.

\paragraph{Stochastic Model.}
Under $\mathcal{H}_0$ (no watermark), we do not assume that the raw token sequence is uniformly random or independent. Instead, we condition on any fixed text $X$, which may exhibit arbitrary linguistic dependencies, and attribute all randomness to the secret key $\mathcal{K}$.
Under the Random Oracle Model (ROM), the block-specific seeds $\mathrm{seed}_j = H(\mathcal{K}, j)$ are independent random variables across blocks. Since the detector-facing symbol observations and corresponding detection events are deterministic functions of the fixed text and $\mathrm{seed}_j$, the block-level detection indicators are mutually independent under $\mathcal{H}_0$.

Thus, the independence assumption does not rely on properties of natural language, but follows from the pseudorandomness induced by the secret key. Under $\mathcal{H}_1$ (watermark present), each block independently suffers symbol errors (from soft embedding and/or adversarial editing) with per-symbol error probability $p_{\mathrm{tot}}\in[0,1]$, and a bounded linear misalignment of at most $s_{\max}$ symbols may be introduced due to insertions or deletions.

\subsection{Na\"{i}ve ``Any-codeword'' Presence Test}
\label{Appendix:finite:any}
Consider the (undesirable) test that declares a watermark if \emph{there exists} any codeword of $C$ within Hamming distance $t$ of the observed block.

\begin{theorem}[FPR of the any-codeword test]
\label{thm:any}
If $t\le \lfloor(d_{\min}-1)/2\rfloor$ so that Hamming balls of radius $t$ around distinct codewords are disjoint, then under $\mathcal{H}_0$ the single-block false-positive probability of the any-codeword test is
\begin{equation}
\mathrm{FPR}_{\textsc{any}} = \frac{|C|V_q(n,t)}{q^n}= q^{k-n}V_q(n,t).
\end{equation}
\end{theorem}

\begin{proof}
Under $\mathcal{H}_0$, the block is uniform on $\Sigma^n$. The event ``within distance $t$ of \emph{some} codeword'' is the disjoint union of the $|C|$ Hamming balls of radius $t$, each of volume $V_q(n,t)$. The probability is therefore $|C|V_q(n,t)/q^n$. 
\end{proof}

\begin{remark}[Binary specialization and magnitude]
For $q=2$, the Hamming ball volume is $V_2(n,t)=\sum_{i=0}^{t}\binom{n}{i}$. 
Even for moderate parameters, this quantity can be large: for BCH-like $(n,t)=(31,7)$,
$V_2=3{,}572{,}224$, yielding
$\mathrm{FPR}_{\textsc{any}} = 2^{k-n} V_2$.
For the default BCH$(31,6,7)$ configuration used in our experiments, this gives
$\mathrm{FPR}_{\textsc{any}} \approx 0.106$, which is unacceptably high.
This motivates the designated-codeword test introduced below.
\end{remark}

\subsection{Designated-Codeword Test}
\label{Appendix:finite:designated}
Our scheme designates exactly one valid codeword per block $j$ via $\mathrm{seed}_j$; only proximity to this designated codeword is considered. The following analysis characterizes the detector-facing $q$-ary symbol sequence under an idealized hash-induced null model. It should not be interpreted as assuming that raw natural-language tokens are uniformly random.

\begin{theorem}[Single-block FPR under designated-codeword test]
\label{thm:single_fpr}
Condition on a fixed unwatermarked text under $\mathcal{H}_0$ and consider the detector-facing $q$-ary block induced by the secret-key hash partition. Under the idealized uniform $q$-ary null model, the single-block false positive rate (FPR) for the designated-codeword test is
\begin{equation}
p_0 = \frac{V_q(n,t)}{q^n}.
\label{eq:p0}
\end{equation}
When detection is performed using a sliding window over
$S=2s_{\max}+1$ bounded linear offsets, the FPR obeys the union bound
\begin{equation}
p_0^{(\mathrm{shift})} \le \min\{1,Sp_0\}.
\label{eq:shift_union}
\end{equation}
If the $S$ shifted decoding events are independent or have negligibly overlapping acceptance regions, then
\begin{equation}
p_0^{(\mathrm{shift})} = 1-(1-p_0)^S = Sp_0 + O(p_0^2).
\label{eq:shift_exact}
\end{equation}
\end{theorem}

\begin{proof}
For a fixed designated codeword $c^{(j)}$, the idealized detector-facing $q$-ary null model assigns equal probability to each vector in $\{0,\ldots,q-1\}^n$. The number of vectors within Hamming radius $t$ of $c^{(j)}$ is $V_q(n,t)$, which gives \eqref{eq:p0}. Searching $S$ shifted framings gives at most $S$ opportunities to enter an acceptance region, yielding the union bound in \eqref{eq:shift_union}. Under the additional independence approximation across shifted decoding events, the complement probabilities multiply, yielding \eqref{eq:shift_exact}.
\end{proof}

\begin{remark}[Entropy bound]
For any $q$, $V_q(n,t)\le q^{\,n H_q(t/n)}$ where $H_q(\cdot)$ is the $q$-ary entropy. Thus
\begin{equation}
p_0 \le q^{-n(1-H_q(t/n))}, 
\qquad 
p_0^{(\mathrm{shift})}\lesssim S\,q^{-n(1-H_q(t/n))}.
\label{eq:entropy_bound}
\end{equation}
This highlights the exponential FPR decay in $n$ at fixed $t/n$.
\end{remark}

\begin{example}[Binary instances]
For $q=2$ and $(n,t)=(31,7)$, which corresponds to our default experimental configuration, $p_0=3,572,224/2^{31}\approx 1.6634 \times 10^{-3}$. With $s_{\max}=5$ ($S=11$), $p_0^{(\mathrm{shift})}\approx 1.81 \times 10^{-2}$ via \eqref{eq:shift_exact}. For $(n,t)=(63,3)$ and $(127,5)$, $p_0\approx 4.52 \times 10^{-15}$ and $1.56 \times 10^{-30}$, respectively.
\end{example}

\subsection{Aggregate FPR with a Match-Ratio Threshold}
\label{Appendix:finite:aggfpr}

We first relate the analysis to Algorithm~\ref{alg:detection}, which evaluates one global offset candidate at a time. Let $X_j^{(s)}$ indicate whether block $j$ matches its designated codeword under global offset candidate $s\in\mathcal{S}$, and define $Y_j=\max_{s\in\mathcal{S}}X_j^{(s)}$. Then,
\[
\max_{s\in\mathcal{S}}\frac{1}{M}\sum_{j=1}^{M}X_j^{(s)}
\le
\frac{1}{M}\sum_{j=1}^{M}\max_{s\in\mathcal{S}}X_j^{(s)}
=
\frac{1}{M}\sum_{j=1}^{M}Y_j .
\]
Thus, the per-block maximization model analyzed below is more permissive than the actual global-offset detector, and its FPR bound is a conservative upper bound for Algorithm~\ref{alg:detection}.

Let $Y_j$ be the indicator that block $j$ matches under the best offset in $\mathcal{S}$ under $\mathcal{H}_0$. Write $p\triangleq p_0^{(\mathrm{shift})}$.

\begin{theorem}[Aggregate FPR under thresholding]
\label{thm:agg_fpr}
Assume $\{Y_j\}_{j=1}^M$ are independent Bernoulli random variables with
$\Pr[Y_j=1]\le p$. Then for any threshold $\theta$ satisfying $p<\theta<1$,
\begin{equation}
\Pr_{\mathcal{H}_0} \bigg[\frac{1}{M}\sum_{j=1}^M Y_j \ge \theta\bigg]
\le
\exp\Big(-M\,D(\theta\Vert p)\Big),
\end{equation}
where $D(a\Vert b)=a\log\frac{a}{b}+(1-a)\log\frac{1-a}{1-b}$ is the Bernoulli KL divergence.
\end{theorem}

\begin{proof}
This is the standard Chernoff (Cram\'{e}r--Chernoff) bound for binomial tails.
\end{proof}

\begin{remark}[Validity of independence assumption]
Theorem~\ref{thm:agg_fpr} relies on the independence of the relaxed block-level matching indicators $Y_j$. We do not assume that the raw text tokens are independent under $\mathcal{H}_0$. Under $\mathcal{H}_0$, the text is independent of the secret key; hence, conditioning on any fixed text, the remaining randomness comes from the block-specific hash seeds $\mathrm{seed}_j=H(\mathcal{K},j)$. Under the Random Oracle Model, these seeds are modeled as independent random variables across distinct block indices, which pseudorandomizes the detector-facing symbol observations across blocks. Therefore, the independence assumption applies to the detector-facing matching indicators induced by randomized partitioning, rather than to the raw natural language sequence itself.
\end{remark}

\begin{remark}[Design implication]
Choosing $\theta\gg p$ makes the aggregate FPR exponentially small in $M$.
In particular, combining \eqref{eq:entropy_bound} and Theorem~\ref{thm:agg_fpr} yields exponential suppression as both $n$ and $M$ increase, at fixed $t/n$ and $S$.
\end{remark}

\begin{corollary}[FPR Correction for Blind Detection]
\label{cor:blind_correction}
In the blind detection protocol (Stage 1), the detector estimates a candidate payload $\hat{m}$ from a message space of size $2^k$ before verification. By applying a union bound over all possible messages, the aggregate false positive rate in the blind setting is bounded by:
\begin{equation}
\mathrm{FPR}_{\mathrm{blind}} 
\le 2^k \cdot \Pr_{\mathcal{H}_0} \bigg[\frac{1}{M}\sum_{j=1}^M Y_j \ge \theta\bigg] 
\le 2^k \cdot \exp\Big(-M\,D(\theta\Vert p)\Big).
\end{equation}
Although the bound increases linearly by the message space size $2^k$, the base probability decays exponentially in $M$ and doubly exponentially in $n$. For typical parameters, the factor $2^k$ is negligible compared to the suppression provided by the Chernoff exponent, ensuring the final FPR remains low.
\end{corollary}

\subsection{Soft Embedding: Symbol Error Induced by $\delta$-Bias}
\label{Appendix:finite:soft}
Let $m\in(0,1)$ denote the \emph{pre-bias} total softmax mass of the green list at a generation step. After applying the logit shift $+\delta$ to the green tokens, the probability that the next token is drawn from the green list is
\begin{equation}
P_{\text{green}}(\delta;m)=\frac{me^{\delta}}{me^{\delta}+(1-m)}=\sigma\big(\mathrm{logit}(m)+\delta\big),
\end{equation}
where $\sigma(u)=1/(1+e^{-u})$ and $\mathrm{logit}(m)=\log(m/(1-m))$.

\begin{theorem}[Per-symbol embedding error in soft mode]
\label{thm:pemb}
When the designated symbol requires sampling from the green list, the per-symbol embedding error probability is
\begin{align}
p_{\mathrm{emb}}(\delta;m) &= 1-P_{\text{green}}(\delta;m)=\frac{1-m}{m\,e^{\delta}+(1-m)}.
\end{align}
In the balanced case $m=\tfrac{1}{2}$, $p_{\mathrm{emb}}(\delta;\tfrac{1}{2})=1-\sigma(\delta)$. For a target $p^* \in(0,1/2)$, it suffices to choose
\begin{equation}
\delta \ge \log\frac{1-p^*}{p^*} - \mathrm{logit}(m)
\end{equation}
to guarantee $p_{\mathrm{emb}}(\delta;m)\le p^*$.
\end{theorem}

\begin{proof}
It is straightforward from the softmax with a uniform logit shift on the green subset. The inequality is obtained by solving $1-P_{\text{green}}(\delta;m)\le p^*$ for $\delta$.
\end{proof}

\begin{example}
For $m=\tfrac{1}{2}$, $\delta\in\{2.0,\,2.5,\,3.0\}$ yields $p_{\mathrm{emb}}\approx\{0.1192,\,0.0759,\,0.0474\}$, respectively.
\end{example}

\subsection{Detection Power under Embedding and Attack Noise}
\label{Appendix:finite:fnr}
Let $p_{\mathrm{att}}\in[0,1]$ be the adversarial symbol error rate within a block (e.g., substitutions after alignment). A conservative union bound gives $p_{\mathrm{tot}}\le p_{\mathrm{emb}}+p_{\mathrm{att}}$.

\begin{theorem}[Single-block success and aggregate FNR]
\label{thm:fnr}
Suppose a block experiences i.i.d.\ symbol errors with probability $p_{\mathrm{tot}}$ and a bounded linear misalignment of at most $s_{\max}$ tokens, so that the correct alignment is included in the sliding-window search. Then the single-block success probability is
\begin{equation}
p_1(n,t,p_{\mathrm{tot}})=\Pr[\mathrm{Bin}(n,p_{\mathrm{tot}})\le t]=\sum_{i=0}^{t}\binom{n}{i}p_{\mathrm{tot}}^{i}(1-p_{\mathrm{tot}})^{n-i}.
\end{equation}
If $Y_j$ are i.i.d.\ indicators of success across blocks under $\mathcal{H}_1$, the aggregate false-negative probability obeys
\begin{equation}
\Pr_{\mathcal{H}_1}\bigg[\frac{1}{M}\sum_{j=1}^{M}Y_j<\theta\bigg]
\le
\exp\Big(-M\,D(\theta\Vert p_1)\Big).
\end{equation}
\end{theorem}

\begin{proof}
Unique decoding succeeds iff the number of symbol errors does not exceed $t$; the binomial tail gives the expression. The Chernoff bound for the lower tail yields the aggregate exponent.
\end{proof}

\begin{example}[Guideline at $(n,t)=(31,7)$]
With $m=\tfrac{1}{2}$ and $\delta=2.5$, the embedding error probability is
$p_{\mathrm{emb}}\approx 0.0759$. 
If the attack-induced bit error probability satisfies $p_{\mathrm{att}}\in[0,0.01]$, then
$p_{\mathrm{tot}}\in[0.0759,0.0859]$, yielding a per-block recovery probability
$p_1\approx 0.79$ to $0.73$. 
For $M=32$ blocks and a detection threshold $\theta=0.5$, the aggregate false negative rate
is exponentially small by Theorem~\ref{thm:fnr}.
\end{example}

\subsection{Alignment Recovery via Sliding Window}
\label{Appendix:finite:shift}

\begin{lemma}[Alignment Recovery]
\label{lem:shift}
Let $\Delta_j$ be the net accumulation of insertions ($+$) and deletions ($-$) prior to block $j$. 
If $|\Delta_j| \le s_{\max}$, then there exists a linear offset $s = \Delta_j \in \mathcal{S}$ such that the window $\mathbf{b}[j\cdot n + s : j\cdot n + s + n]$ perfectly aligns with the embedded codeword boundaries.
Under this alignment, decoding succeeds provided the internal symbol errors (substitution) are within radius $t$.
\end{lemma}

\begin{proof}
Since the detector searches all linear offsets $s \in [-s_{\max}, s_{\max}]$, the set $\mathcal{S}$ includes the true shift $\Delta_j$. 
Selecting this window eliminates the framing error caused by insertion/deletion, reducing the problem to standard substitution error correction. 
Thus, unique decoding is guaranteed if the substitution noise is $\le t$.
\end{proof}

\begin{remark}[Linear vs. Circular]
Unlike prior works that model edits as circular shifts within a fixed block, our analysis accounts for the linear displacement of the bit stream. The sliding window search ($O(s_{\max} \cdot n)$ complexity) effectively neutralizes the synchronization drift.
\end{remark}

\subsection{Parameter selection via Entropy Bounds}
\label{Appendix:finite:params}
The entropy bound in \eqref{eq:entropy_bound}, together with Theorems~\ref{thm:agg_fpr} and \ref{thm:fnr}, suggests a simple two-sided design: pick $(n,t,s_{\max},\theta,M,\delta)$ so that
\begin{align}
\underbrace{\exp\Big(-M\,D(\theta\Vert p_0^{(\mathrm{shift})})\Big)}_{\text{FPR target }\alpha}&\le\alpha,
\qquad
p_0^{(\mathrm{shift})}\approx 1-(1-p_0)^S, p_0\le q^{-n(1-H_q(t/n))},\\
\underbrace{\exp\Big(-M\,D(\theta\Vert p_1)\Big)}_{\text{FNR target }\beta}&\le\beta,
\qquad
p_1=\Pr[\mathrm{Bin}(n,p_{\mathrm{tot}})\le t], p_{\mathrm{tot}}\lesssim p_{\mathrm{emb}}(\delta;m)+p_{\mathrm{att}}.
\end{align}

\begin{remark}[Balanced operation]
A convenient choice is to set $\theta$ near the Chernoff intersection that equalizes exponents $D(\theta\Vert p_0^{(\mathrm{shift})})\approx D(\theta\Vert p_1)$, and to tune $\delta$ to keep $p_{\mathrm{tot}}<t/n$ so that $p_1$ remains bounded away from $1/2$.
\end{remark}

\subsection{Generalizations}
\label{Appendix:finite:general}
\begin{proposition}[Direct $q$-ary extension]
\label{prop:qary}
All the bounds above hold verbatim for $q>2$ with $V_q(n,t)$ from \eqref{eq:Vq}; in particular,
\[
p_0=\frac{V_q(n,t)}{q^n},\qquad 
p_0^{(\mathrm{shift})}\le \min\{1,S\,p_0\},\qquad
\Pr_{\mathcal{H}_0}\big[\mathrm{match\ ratio}\ge \theta\big]\le e^{-M D(\theta\Vert p_0^{(\mathrm{shift})})}.
\]
\end{proposition}

\begin{proof}
Identical combinatorial counting applies because the unique-decoding radius $t$ depends only on $d_{\min}$ and the metric, not on the field size beyond the volume $V_q(n,t)$.
\end{proof}

\section{Supplementary Experiments}
\label{Appendix:supp_exp}

\begin{figure*}[t]
  \centering
  \begin{minipage}[t]{0.53\textwidth}
    \centering
    \includegraphics[width=\linewidth]{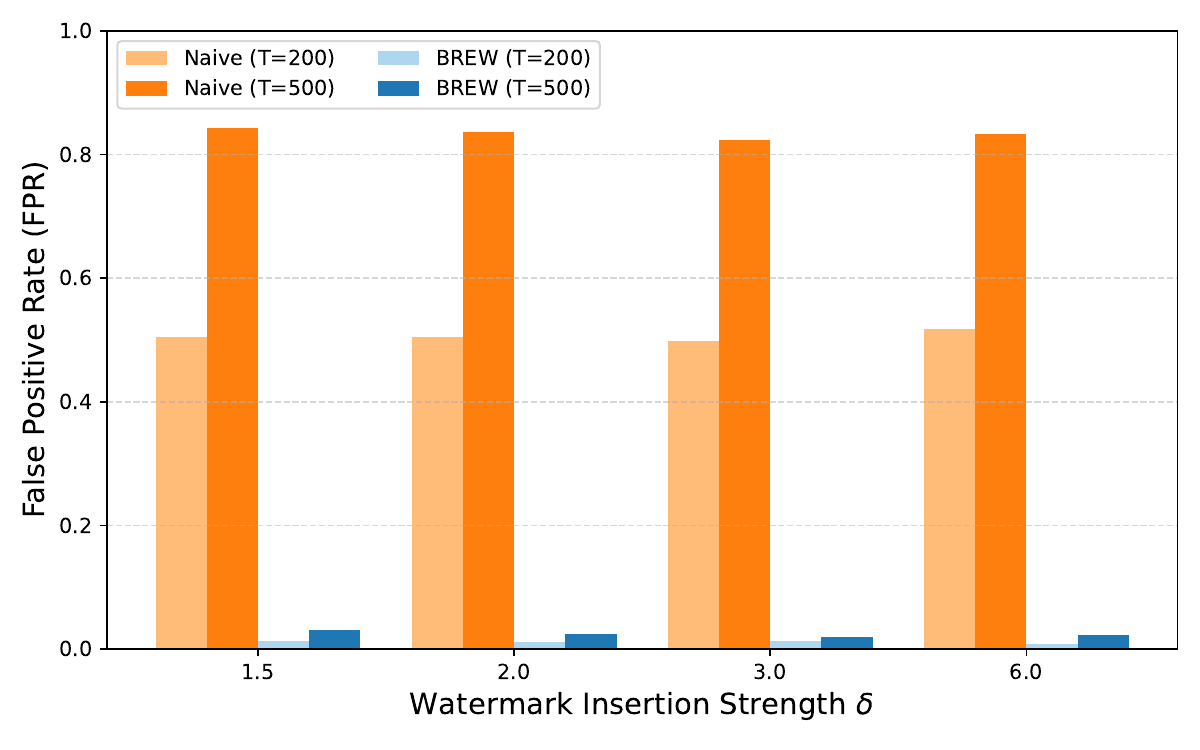}
    \captionof{figure}{False positive rate (FPR) across insertion strengths~$\delta$ under the clean setting.}
    \label{fig:fpr_clean}
  \end{minipage}
  \hfill
  \begin{minipage}[t]{0.43\textwidth}
    \centering
    \includegraphics[width=\linewidth]{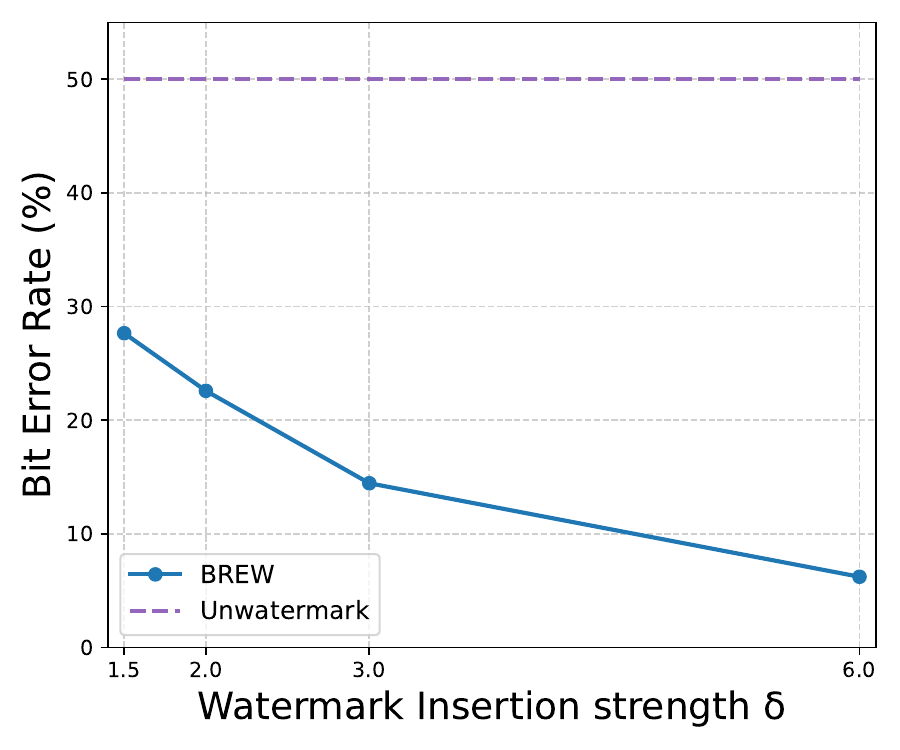}
    \captionof{figure}{Average bit error rate (BER) as a function of watermark insertion strength~$\delta$.}
    \label{fig:ber_delta}
  \end{minipage}
\end{figure*}

\subsection{Detection performance without attacks}
\label{Appendix:no_attack}


This experiment evaluates detection performance of watermarking techniques in clean environments without adversarial attacks, focusing on how watermark insertion strength \(\delta\) and text length \(T\) influence reliability. Figure~\ref{fig:fpr_clean} shows that \emph{Na\"{i}ve-Ours} suffers from consistently high FPR across all \(\delta\) values, failing to distinguish watermarked from unwatermarked text. In contrast, \emph{BREW} maintains FPR close to zero regardless of \(\delta\), demonstrating the effectiveness of structured decoding. For example, at \(\delta{=}3.0\) with \(T{=}200\), \emph{Na\"{i}ve} attains \(\text{TPR}{=}1.000\) but \(\text{FPR}{=}0.499\), whereas \emph{BREW} achieves a comparable \(\text{TPR}{=}0.987\) with \(\text{FPR}{=}0.013\). Complete numerical results across all settings are reported in Table~\ref{table:No_attack}.

\begin{table}[t]
\centering
\caption{Detection performance without attacks (\(T=200\) vs. \(T=500\))} 
\begin{tabular}{ll*{4}{c}*{4}{c}}
\toprule
Scheme & Setting & \multicolumn{4}{c}{T200} & \multicolumn{4}{c}{T500} \\  
\cmidrule(lr){3-6} \cmidrule(lr){7-10}  
 &  & TPR & FPR & Precision & F1 & TPR & FPR & Precision & F1 \\  
\midrule 
 Na\"{i}ve-Ours (soft)  
 & $\delta=1.5$ & 0.950 & 0.504 & 0.6534 & 0.7742 & 1.000 & 0.843 & 0.5426 & 0.7035 \\  
 & $\delta=2.0$ & 0.997 & 0.504 & 0.6642 & 0.7973 & 1.000 & 0.837 & 0.5444 & 0.7050 \\  
 & $\delta=3.0$ & 1.000 & 0.499 & 0.6671 & 0.8003 & 1.000 & 0.824 & 0.5482 & 0.7082 \\  
 & $\delta=6.0$ & 1.000 & 0.518 & 0.6588 & 0.7943 & 1.000 & 0.833 & 0.5456 & 0.7060 \\ 
\midrule  
Na\"{i}ve-Ours (hard)  
 &  & 1.000 & 0.475 & 0.6780 & 0.8081 & 1.000 & 0.861 & 0.5373 & 0.6991 \\ 
\midrule  
 BREW (soft)  
 & $\delta=1.5$ & 0.906 & 0.013 & 0.9859 & 0.9442 & 0.975 & 0.031 & 0.9692 & 0.9721 \\  
 & $\delta=2.0$ & 0.977 & 0.011 & 0.9889 & 0.9829 & 0.987 & 0.024 & 0.9763 & 0.9816 \\  
 & $\delta=3.0$ & 0.987 & 0.013 & 0.9870 & 0.9870 & 0.986 & 0.019 & 0.9811 & 0.9835 \\  
 & $\delta=6.0$ & 1.000 & 0.011 & 0.9891 & 0.9945 & 0.999 & 0.087 & 0.9328 & 0.9648 \\ 
\midrule  
BREW (hard)  
 & & 0.982 & 0.013 & 0.9869 & 0.9844 & 0.974 & 0.032 & 0.9682 & 0.9711 \\  
\bottomrule  
\end{tabular}  
\label{table:No_attack}  
\end{table}

\subsection{Bit Error Rate Analysis by Watermark Insertion Strength \(\delta\)}
\label{Appendix:BER}


This experiment evaluates the effect of watermark insertion strength \(\delta\) on bit-level codeword reconstruction. As shown in Figure~\ref{fig:ber_delta}, unwatermarked text consistently exhibits about 50\% BER, which corresponds to random guessing. In contrast, watermarked text yields significantly lower BERs, with the error rate steadily decreasing as \(\delta\) increases. These results demonstrate that stronger watermark insertion improves the reliability of codeword reconstruction, whereas smaller values of \(\delta\) result in BERs closer to random noise, rendering detection more difficult.

\subsection{Text Quality Under Watermarking}
\label{Appendix:Text_Quality}

\begin{table}[t]
  \caption{Comparison of text quality across watermarking schemes.}
  \label{table:text_quality}
  \begin{center}
  \begin{tabular}{llccc}
    \toprule
    Scheme & $\delta$ & PPL ($\downarrow$) & BLEU ($\uparrow$) & BERTScore ($\uparrow$) \\
    \midrule
    Unwatermarked & -- & 15.51 & 31.81 & 0.8201 \\
    \midrule
    MPAC
      & 1.5 & 10.80 & 29.86 & 0.8136 \\
      & 2.0 & 11.24 & 28.48 & 0.8069 \\
      & 3.0 & 13.87 & 22.22 & 0.7771 \\
      & 6.0 & 24.09 & 6.43 & 0.6279 \\
    \midrule
    \cite{qu2025provably}
      & 1.5 & 11.96 & 30.99 & 0.8132 \\
      & 2.0 & 13.02 & 26.44 & 0.7996 \\
      & 3.0 & 15.77 & 22.31 & 0.7740 \\
      & 6.0 & 23.45 & 10.14 & 0.6688 \\
    \midrule
    BREW (soft)
      & 1.5 & 11.92 & 29.14 & 0.8132 \\
      & 2.0 & 13.24 & 27.78 & 0.8082 \\
      & 3.0 & 15.81 & 20.13 & 0.7738 \\
      & 6.0 & 24.46 & 10.16 & 0.6864 \\
    \midrule
    BREW (hard) & -- & 30.75 & 6.69 & 0.6312 \\
    \bottomrule
  \end{tabular}
  \end{center}
\end{table}

Table~\ref{table:text_quality} reports text quality metrics under different watermarking schemes.
As expected, unwatermarked text achieves the best overall quality, with the highest BLEU~\cite{papineni2002} and BERTScore~\cite{zhang2020bertscore}.
Across all watermarking methods, increasing the embedding strength $\delta$ leads to higher perplexity (PPL) and lower BLEU and BERTScore, reflecting the inherent trade-off between watermark robustness and text quality.

Among watermarking approaches, BREW (soft) and \cite{qu2025provably} exhibit broadly comparable quality trends, while MPAC~\cite{yoo2024advancing} degrades more rapidly as $\delta$ increases, showing pronounced drops in BLEU and BERTScore.
At $\delta=2.0$, BREW (soft) achieves BLEU $=27.78$ and BERTScore $=0.8082$, compared to $26.44$ and $0.7996$ for \cite{qu2025provably}.
Under the strongest watermarking setting ($\delta=6.0$), BREW (soft) yields slightly higher PPL than MPAC (24.46 vs.\ 24.09) but substantially higher BLEU (10.16 vs.\ 6.43) and BERTScore (0.6864 vs.\ 0.6279).
In contrast, the hard variant of BREW incurs substantial quality degradation, confirming that soft watermarking provides a more favorable balance between robustness and text quality.

\subsection{Ablation Study: Effect of Designated Verification and Window-Shifting}
\label{Appendix:Ablation}

This ablation study analyzes the contribution of the two core components of BREW: \emph{designated codeword verification} and \emph{window-shifting detection}. All results are reported under a 10\% token-increasing synonym substitution attack on the C4 dataset using OPT-1.3B with watermark strength $\delta=6$.
We compare three detector variants:
(1) \textbf{Designated-only}, which verifies only designated codewords without window-shifting;
(2) \textbf{Shift-only}, which applies window-shifting but accepts any decodable codeword; and
(3) \textbf{Both}, which combines designated verification with window-shifting (full BREW).
The results confirm that designated codeword verification is essential for suppressing false positives, while window-shifting detection primarily improves recall under insertion-induced misalignment. Combining both components yields the best overall TPR--FPR trade-off.

\begin{table}[t]
    \centering
    \label{tab:ablation_study}
    \caption{Ablation study comparing designated-only, shift-only, and full BREW detectors under a 10\% token-increasing synonym substitution attack (C4, OPT-1.3B, $\delta=6$). \textbf{Designated-only} verifies designated codewords without window-shifting, \textbf{shift-only} applies window-shifting but accepts any decodable codeword, and \textbf{both} corresponds to the full BREW detector.}
    \begin{tabular}{llcccccccc}
    \toprule
     model 
    &  $s_{\max}$ 
    & \multicolumn{4}{c}{T200} & \multicolumn{4}{c}{T500} \\ \cmidrule(lr){3-6} \cmidrule(lr){7-10} 
     &  & TPR & FPR & Precision & F1 & TPR & FPR & Precision & F1 \\ 
     \midrule
    designated-only & - & 0.355 & 0.005 & 0.9861 & 0.5221 & 0.430 & 0.015 & 0.9663 & 0.5952 \\ \midrule
     shift-only 
     & 1 & 0.205 & 0.030 & 0.8723 & 0.3320 & 0.330 & 0.040 & 0.8919 & 0.4818 \\
     & 3 & 0.540 & 0.035 & 0.9391 & 0.6857 & 0.745 & 0.130 & 0.8514 & 0.7947 \\
     & 5 & 0.725 & 0.140 & 0.8382 & 0.7775 & 0.880 & 0.190 & 0.8224 & 0.8502 \\ \midrule
     both 
     & 1 & 0.500 & 0.035 & 0.9346 & 0.6515 & 0.630 & 0.075 & 0.8936 & 0.7390 \\
     & 3 & 0.730 & 0.050 & 0.9346 & 0.8202 & 0.815 & 0.105 & 0.8859 & 0.8490 \\
     & 5 & 0.815 & 0.110 & 0.8810 & 0.8468 & 0.960 & 0.265 & 0.7837 & 0.8629 \\ \bottomrule
    \end{tabular}
\end{table}

\subsection{Sensitivity to Window-Shift Range and $\delta$}
\label{Appendix:Window_shift_Range_delta}

\begin{figure}[t]
  \begin{center}
    \includegraphics[width=\columnwidth]{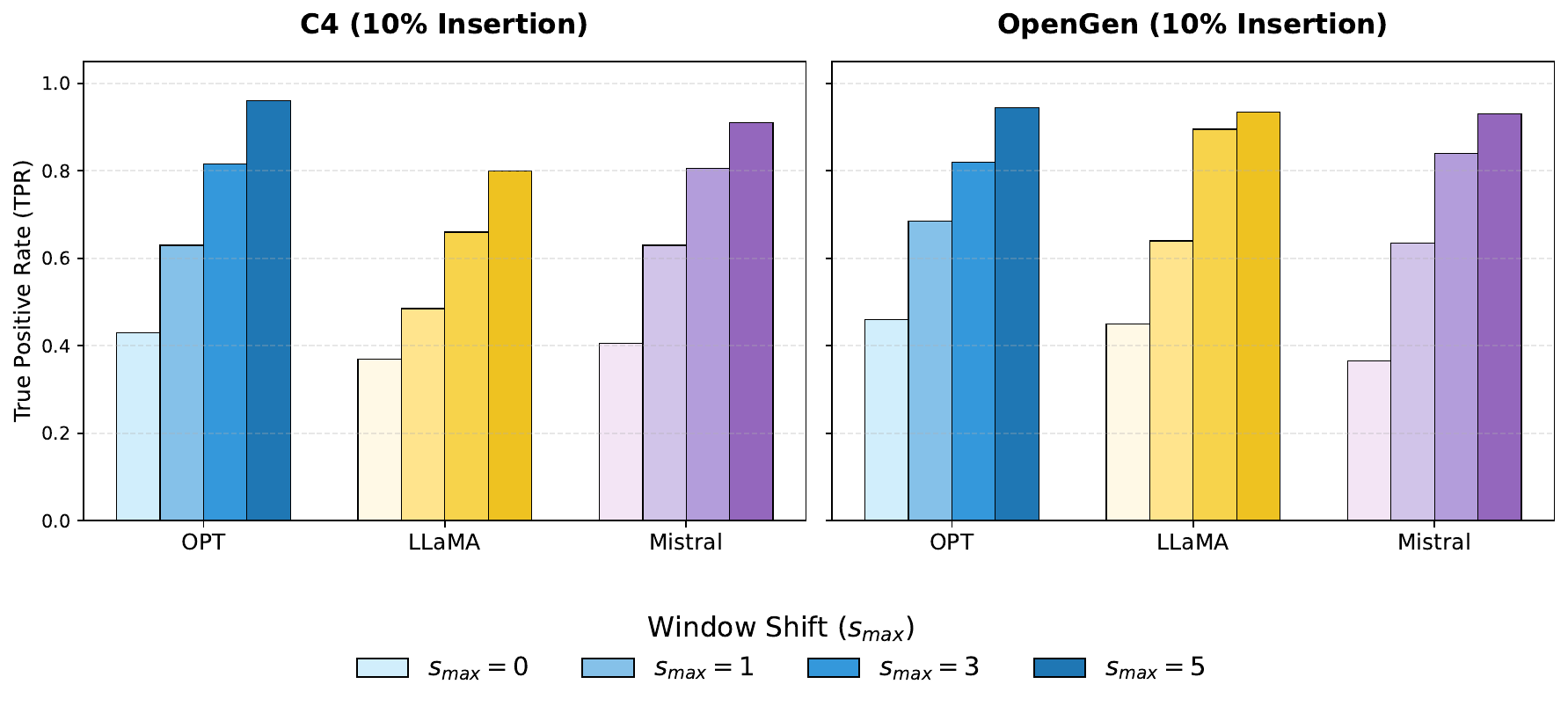}
    \caption{Effect of the window-shift range $s_{\max}$ on the true positive rate (TPR) under a fixed 10\% insertion attack. Increasing $s_{\max}$ consistently improves TPR on both C4 and OpenGen datasets, demonstrating that window-shifting effectively compensates for insertion-induced token-level misalignment.}
    \label{fig:S_max}
  \end{center}
\end{figure}
\begin{figure}[t]
  \begin{center}
    \includegraphics[width=\columnwidth]{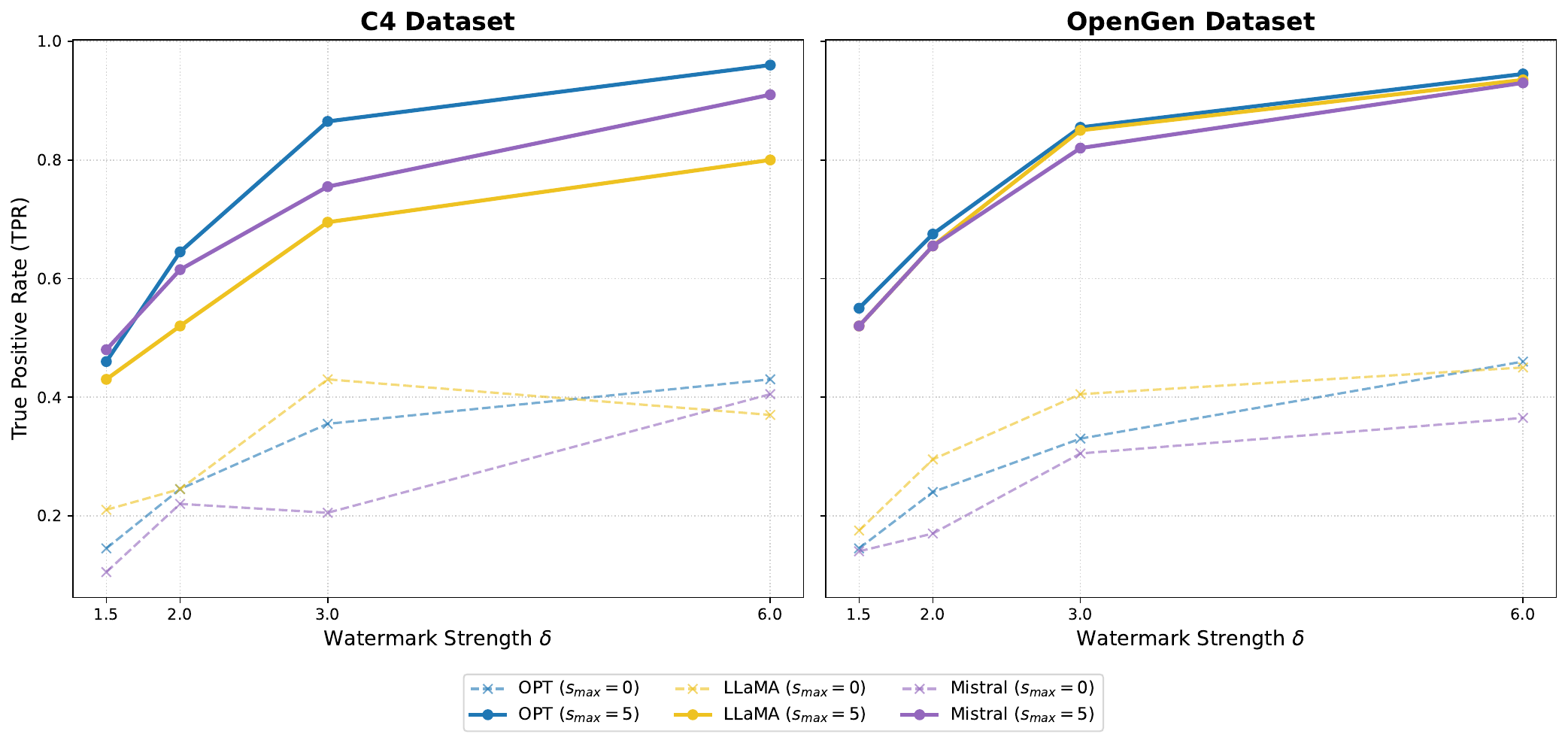}
    \caption{Sensitivity of the true positive rate (TPR) to the watermark embedding strength $\delta$ under a 10\% insertion attack. TPR increases monotonically with larger $\delta$ across all model backends, with substantially stronger gains when window-shifting is enabled ($s_{\max}=5$), highlighting the complementary role of watermark strength and alignment recovery.}
    \label{fig:delta_comparsion}
  \end{center}
\end{figure}

This subsection analyzes the sensitivity of detection performance to the window-shift range $s_{\max}$ and the watermark embedding strength $\delta$ under token-level attacks, with all results reported on texts of length $T=500$. Token-altering attacks such as insertion and deletion disrupt the alignment between embedded codeword boundaries and the observed token sequence, making reliable detection challenging.

Under a fixed 10\% insertion attack, increasing $s_{\max}$ consistently improves the true positive rate (TPR) on both C4 and OpenGen datasets (Figure~\ref{fig:S_max}), indicating that window-shifting effectively compensates for insertion-induced misalignment. The gains saturate at moderate shift ranges, suggesting that a small alignment budget is sufficient in practice.

Under the same setting, increasing the watermark embedding strength $\delta$ further improves TPR across all backbone models (Figure~\ref{fig:delta_comparsion}), with substantially stronger improvements observed when window-shifting is enabled ($s_{\max}=5$). This highlights the complementary roles of alignment flexibility provided by $s_{\max}$ and watermark signal strength controlled by $\delta$.

We note that increasing $s_{\max}$ also slightly increases the false positive rate due to the expanded search space; however, this trade-off is systematically controlled by the designated-codeword test. Detailed ROC figures for all backbone models are provided in Appendix~\ref{Appendix:synsub_figures}, while full numerical detection tables for OPT-1.3B are reported in Appendix~\ref{app:opt_full_results}.

\subsection{Computational Cost of Window-Shift Detection}
\label{Appendix:Computational_Cost}
The window-shift procedure is applied only during detection and therefore incurs minimal computational overhead. As shown in Table~\ref{table:inference_time}, inference time remains well below one second across all settings, even for larger codeword lengths $n$ and higher shift budgets $s_{\max}$. In the worst case ($T=500$, $n=63$, $s_{\max}=5$), detection completes in under 0.6 seconds, demonstrating that the proposed detection framework is computationally lightweight and practical for real-world deployment.

\begin{table}[t]
  \caption{Inference time (seconds) for varying $s_{\max}$ and codeword lengths $n$.}
  \label{table:inference_time}
  \begin{center}
  \begin{tabular}{cc cccc}
    \toprule
    \multicolumn{2}{c}{\textbf{Setting}} & \multicolumn{4}{c}{\textbf{Inference time (sec)}} \\
    \cmidrule(lr){1-2}\cmidrule(lr){3-6}
    \textbf{T} & \textbf{$n$}
    & $\boldsymbol{s_{\max}=0}$ & $\boldsymbol{s_{\max}=1}$
    & $\boldsymbol{s_{\max}=3}$ & $\boldsymbol{s_{\max}=5}$ \\
    \midrule
    200
     & 15 & 0.0252 & 0.0474 & 0.0861 & 0.1262 \\
     & 31 & 0.0270 & 0.0526 & 0.1146 & 0.1682 \\
     & 63 & 0.0270 & 0.0607 & 0.1345 & 0.2017 \\
    \midrule
    500
     & 15 & 0.0394 & 0.0931 & 0.1698 & 0.2741 \\
     & 31 & 0.0489 & 0.1127 & 0.2577 & 0.3788 \\
     & 63 & 0.0552 & 0.1540 & 0.3431 & 0.5342 \\
    \bottomrule
  \end{tabular}
  \end{center}
\end{table}

\subsection{Effect of Threshold Increase}
\label{Appendix:threshold_two_codewords}

\begin{figure}[t]
  \begin{center}
    \includegraphics[width=\columnwidth]{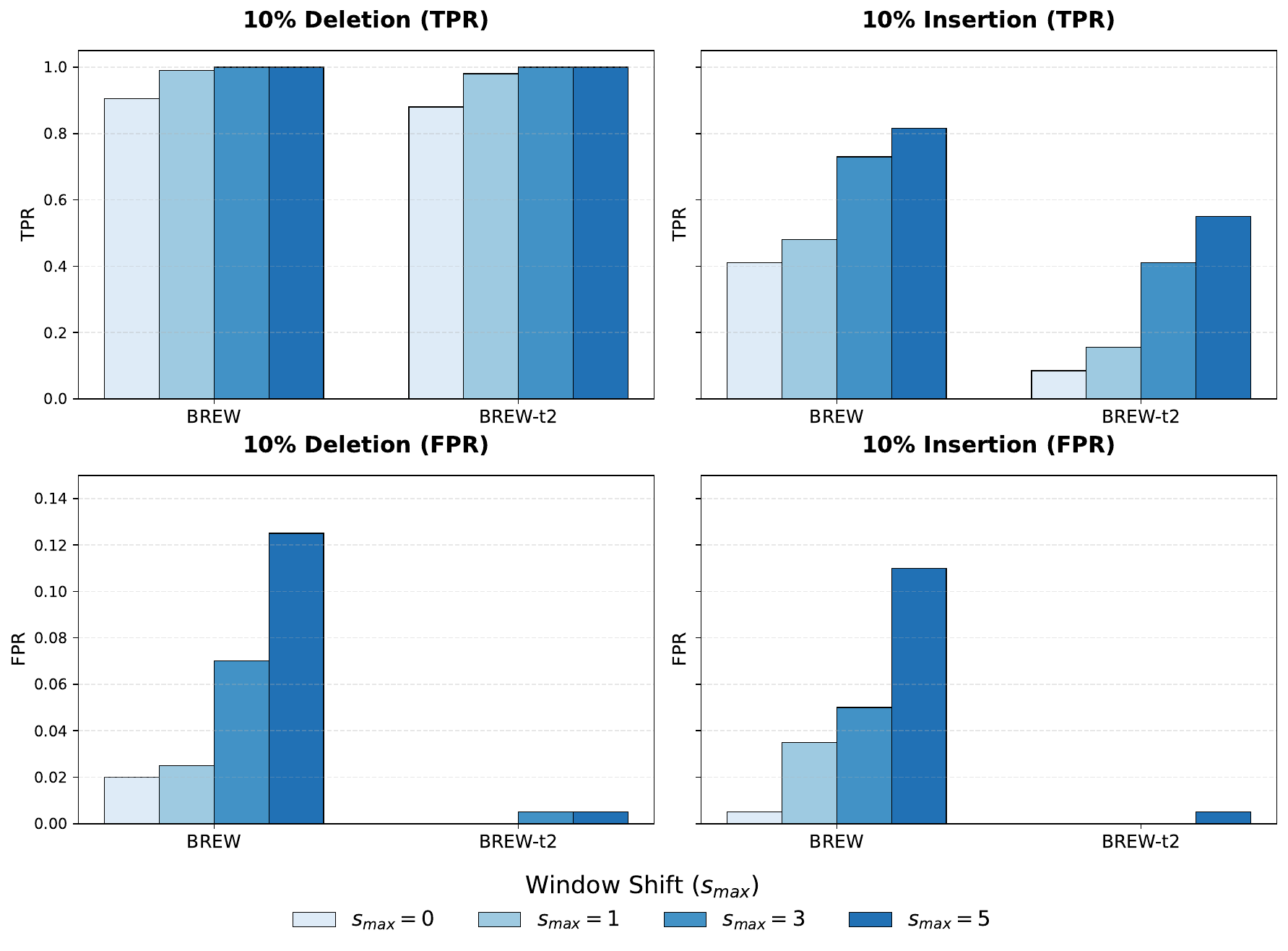}
    \caption{Effect of increasing the detection threshold from one matched codeword (BREW) to two matched codewords (BREW-t2) under \textbf{token-altering synonym substitution attacks} (deletion-like and insertion-like) at a 10\% rate on C4 using OPT-1.3B ($T=200$, $\delta=6$).}
    \label{fig:Threshold}
  \end{center}
\end{figure}

This experiment examines the effect of increasing the detection threshold, i.e., requiring a larger number of matched codewords to declare watermark presence, under token-altering synonym substitution attacks. We compare the default BREW detector, which declares watermark presence if at least one designated codeword is recovered, with a stricter variant (BREW-t2) that requires at least two matched codewords. Figure~\ref{fig:Threshold} illustrates the resulting trade-offs under 10\% deletion-like and insertion-like substitution on the C4 dataset using OPT-1.3B with $\delta=6$ and $T=200$. Increasing the detection threshold consistently suppresses the false positive rate (FPR), with BREW-t2 maintaining near-zero FPR across all window-shift ranges $s_{\max}$, at the cost of reduced true positive rate (TPR) as some watermarked texts fail to recover multiple intact codewords after token-level perturbations. Table~\ref{tab:integrated_tpr_fpr} provides a comprehensive quantitative comparison across different watermark strengths $\delta$ and window-shift ranges. Overall, stricter thresholds offer a conservative operating mode that prioritizes false-positive control under severe token misalignment; accordingly, we adopt the single-codeword threshold as the default setting in the main experiments and report the two-codeword threshold as an optional configuration for applications requiring extremely low FPR.

\begin{table}[t]
\centering
\caption{Comprehensive performance evaluation of BREW and BREW-t2 under Deletion and Insertion attacks on C4 dataset (OPT-1.3B, $T=200$). BREW-t2 consistently maintains near-zero FPR while TPR scales with $\delta$ and $s_{\max}$.}
\label{tab:integrated_tpr_fpr}
\begin{tabular}{lcc | cc | cc | cc | cc}
\toprule
\multicolumn{3}{c|}{\textbf{Setting}} & \multicolumn{4}{c|}{\textbf{Deletion Attacks}} & \multicolumn{4}{c}{\textbf{Insertion Attacks}} \\
\cmidrule(lr){1-3}\cmidrule(lr){4-7}\cmidrule(lr){8-11}
 & & & \multicolumn{2}{c|}{\textbf{5\% Rate}} & \multicolumn{2}{c|}{\textbf{10\% Rate}} & \multicolumn{2}{c|}{\textbf{5\% Rate}} & \multicolumn{2}{c}{\textbf{10\% Rate}} \\
\cmidrule(lr){4-5}\cmidrule(lr){6-7}\cmidrule(lr){8-9}\cmidrule(lr){10-11}
\textbf{Model} & $\delta$ & $s_{\max}$ & \textbf{TPR} & \textbf{FPR} & \textbf{TPR} & \textbf{FPR} & \textbf{TPR} & \textbf{FPR} & \textbf{TPR} & \textbf{FPR} \\
\midrule
BREW & 1.5 & 0 & 0.810 & 0.000 & 0.750 & 0.005 & 0.305 & 0.010 & 0.110 & 0.025 \\
        &     & 1 & 0.895 & 0.025 & 0.900 & 0.020 & 0.380 & 0.025 & 0.185 & 0.040 \\
        &     & 3 & 0.900 & 0.095 & 0.905 & 0.055 & 0.515 & 0.050 & 0.195 & 0.050 \\
        &     & 5 & 0.915 & 0.120 & 0.910 & 0.100 & 0.565 & 0.075 & 0.405 & 0.085 \\
\cmidrule(lr){2-11}
        & 2.0 & 0 & 0.905 & 0.005 & 0.875 & 0.020 & 0.540 & 0.020 & 0.225 & 0.020 \\
        &     & 1 & 0.985 & 0.025 & 0.985 & 0.030 & 0.810 & 0.025 & 0.310 & 0.055 \\
        &     & 3 & 0.995 & 0.045 & 0.995 & 0.080 & 0.695 & 0.090 & 0.380 & 0.080 \\
        &     & 5 & 0.995 & 0.080 & 0.995 & 0.075 & 0.885 & 0.120 & 0.480 & 0.115 \\
\cmidrule(lr){2-11}
        & 3.0 & 0 & 0.920 & 0.005 & 0.895 & 0.010 & 0.630 & 0.020 & 0.390 & 0.010 \\
        &     & 1 & 0.995 & 0.030 & 1.000 & 0.010 & 0.735 & 0.050 & 0.505 & 0.020 \\
        &     & 3 & 1.000 & 0.055 & 1.000 & 0.045 & 0.840 & 0.070 & 0.645 & 0.070 \\
        &     & 5 & 1.000 & 0.090 & 1.000 & 0.105 & 0.930 & 0.085 & 0.710 & 0.110 \\
\cmidrule(lr){2-11}
        & 6.0 & 0 & 0.920 & 0.000 & 0.905 & 0.020 & 0.590 & 0.020 & 0.410 & 0.005 \\
        &     & 1 & 1.000 & 0.045 & 0.990 & 0.025 & 0.630 & 0.020 & 0.480 & 0.035 \\
        &     & 3 & 1.000 & 0.070 & 1.000 & 0.070 & 0.950 & 0.080 & 0.730 & 0.050 \\
        &     & 5 & 1.000 & 0.080 & 1.000 & 0.125 & 0.935 & 0.100 & 0.815 & 0.110 \\
\midrule
BREW-t2 & 1.5 & 0 & 0.515 & 0.000 & 0.525 & 0.000 & 0.060 & 0.000 & 0.005 & 0.000 \\
           &     & 1 & 0.685 & 0.000 & 0.585 & 0.000 & 0.125 & 0.000 & 0.050 & 0.000 \\
           &     & 3 & 0.740 & 0.000 & 0.615 & 0.000 & 0.260 & 0.005 & 0.065 & 0.000 \\
           &     & 5 & 0.700 & 0.005 & 0.695 & 0.000 & 0.260 & 0.005 & 0.055 & 0.000 \\
\cmidrule(lr){2-11}
           & 2.0 & 0 & 0.770 & 0.000 & 0.735 & 0.000 & 0.165 & 0.000 & 0.040 & 0.000 \\
           &     & 1 & 0.955 & 0.000 & 0.925 & 0.000 & 0.230 & 0.000 & 0.085 & 0.000 \\
           &     & 3 & 0.930 & 0.000 & 0.960 & 0.005 & 0.460 & 0.000 & 0.130 & 0.000 \\
           &     & 5 & 0.930 & 0.000 & 0.945 & 0.015 & 0.540 & 0.010 & 0.145 & 0.005 \\
\cmidrule(lr){2-11}
           & 3.0 & 0 & 0.880 & 0.000 & 0.815 & 0.000 & 0.250 & 0.005 & 0.075 & 0.000 \\
           &     & 1 & 0.995 & 0.000 & 0.980 & 0.000 & 0.445 & 0.005 & 0.115 & 0.000 \\
           &     & 3 & 1.000 & 0.000 & 0.995 & 0.010 & 0.675 & 0.000 & 0.280 & 0.000 \\
           &     & 5 & 1.000 & 0.000 & 1.000 & 0.015 & 0.770 & 0.000 & 0.370 & 0.005 \\
\cmidrule(lr){2-11}
           & 6.0 & 0 & 0.885 & 0.000 & 0.880 & 0.000 & 0.255 & 0.000 & 0.085 & 0.000 \\
           &     & 1 & 0.995 & 0.000 & 0.980 & 0.000 & 0.455 & 0.000 & 0.155 & 0.000 \\
           &     & 3 & 1.000 & 0.000 & 1.000 & 0.005 & 0.730 & 0.000 & 0.410 & 0.000 \\
           &     & 5 & 1.000 & 0.005 & 1.000 & 0.005 & 0.890 & 0.005 & 0.550 & 0.005 \\
\bottomrule
\end{tabular}
\end{table}

\subsection{Effect of Codeword Parameters on Detection Performance}
\label{Appendix:codeword_para}


\begin{table*}[t]
  \caption{Detection performance (TPR/FPR) of different BCH codeword configurations at $\delta=3$ and $s_{\max}=5$ under token-altering synonym substitution attacks ($T=200$).}
  \label{tab:bch_summary_delta3_s5}
  \begin{center}
    \begin{small}
      \setlength{\tabcolsep}{4pt}
      \begin{tabular}{l ccc ccc ccc ccc}
        \toprule
        & \multicolumn{3}{c}{\textbf{5\% Deletion}}
        & \multicolumn{3}{c}{\textbf{10\% Deletion}}
        & \multicolumn{3}{c}{\textbf{5\% Insertion}}
        & \multicolumn{3}{c}{\textbf{10\% Insertion}} \\
        \cmidrule(lr){2-4}\cmidrule(lr){5-7}\cmidrule(lr){8-10}\cmidrule(lr){11-13}
        & \textbf{(15,5,3)} & \textbf{(31,6,7)} & \textbf{(63,7,15)}
        & \textbf{(15,5,3)} & \textbf{(31,6,7)} & \textbf{(63,7,15)}
        & \textbf{(15,5,3)} & \textbf{(31,6,7)} & \textbf{(63,7,15)}
        & \textbf{(15,5,3)} & \textbf{(31,6,7)} & \textbf{(63,7,15)} \\
        \midrule
        \textbf{TPR}
        & 1.000 & 1.000 & 1.000
        & 1.000 & 1.000 & 1.000
        & 0.990 & 0.930 & 0.710
        & 1.000 & 0.710 & 0.305 \\
        \textbf{FPR}
        & 0.955 & 0.090 & 0.000
        & 0.930 & 0.105 & 0.000
        & 0.945 & 0.085 & 0.000
        & 0.930 & 0.110 & 0.005 \\
        \bottomrule
      \end{tabular}
    \end{small}
  \end{center}
\end{table*}

We compared three BCH codeword configurations at $T=200$: a short codeword $(n=15,k=5,t=3)$, a medium codeword $(n=31,k=6,t=7)$, and a long codeword $(n=63,k=7,t=15)$. Table~\ref{tab:bch_summary_delta3_s5} summarizes their detection performance at $\delta=3$ and $s_{\max}=5$ under token-altering synonym substitution attacks. The short codeword achieves near-perfect TPR across all settings, but incurs extremely high FPR (e.g., under 10\% deletion, FPR $=0.930$), indicating over-sensitivity to noise. In contrast, the long codeword consistently suppresses FPR to near zero, but suffers from severe TPR degradation under insertion attacks (e.g., TPR drops to $0.305$ under 10\% insertion). The medium codeword strikes a favorable balance between the two extremes, maintaining high TPR while keeping FPR at a moderate level across both deletion and insertion scenarios. Based on this trade-off, we adopt the medium codeword configuration $(n=31,k=6,t=7)$ as the default setting in our main experiments.

\subsection{Message-Level Exact Match Rate}
\label{app:match_rate}

The message-level exact match rate measures the fraction of texts for which the full embedded multi-bit payload is recovered exactly. To avoid overestimating payload recovery, match rate is computed using distinct correctly recovered codewords, without counting duplicate occurrences of the same codeword multiple times.

\begin{table}[t]
\centering
\caption{Message-level exact match rate of BREW at $\delta=6.0$.}
\label{tab:exact_match_rate}
\begin{tabular}{lccc}
\toprule
Attack Scenario & TPR & Match Rate & $\Delta$ (TPR$-$Match Rate) \\ 
\midrule
No Attack & 1.000 & 0.990 & 0.010 \\
Substitution (10\%) & 0.990 & 0.985 & 0.005 \\
Paraphrase & 0.840 & 0.760 & 0.080 \\
\bottomrule
\end{tabular}
\end{table}

As shown in Table~\ref{tab:exact_match_rate}, the gap between TPR and Match Rate is
small under no attack and substitution, indicating that successful detection is
closely aligned with exact payload recovery. Under paraphrasing, BREW still
achieves a Match Rate of $0.760$, demonstrating substantial multi-bit recovery
capability under semantic perturbations.

\subsection{Additional Results for Synonym Substitution Attacks}
\label{Appendix:synsub_figures}

\begin{figure}[t]
  \centering
  \includegraphics[
    width=\columnwidth,
    height=0.9\textheight,
    keepaspectratio
  ]{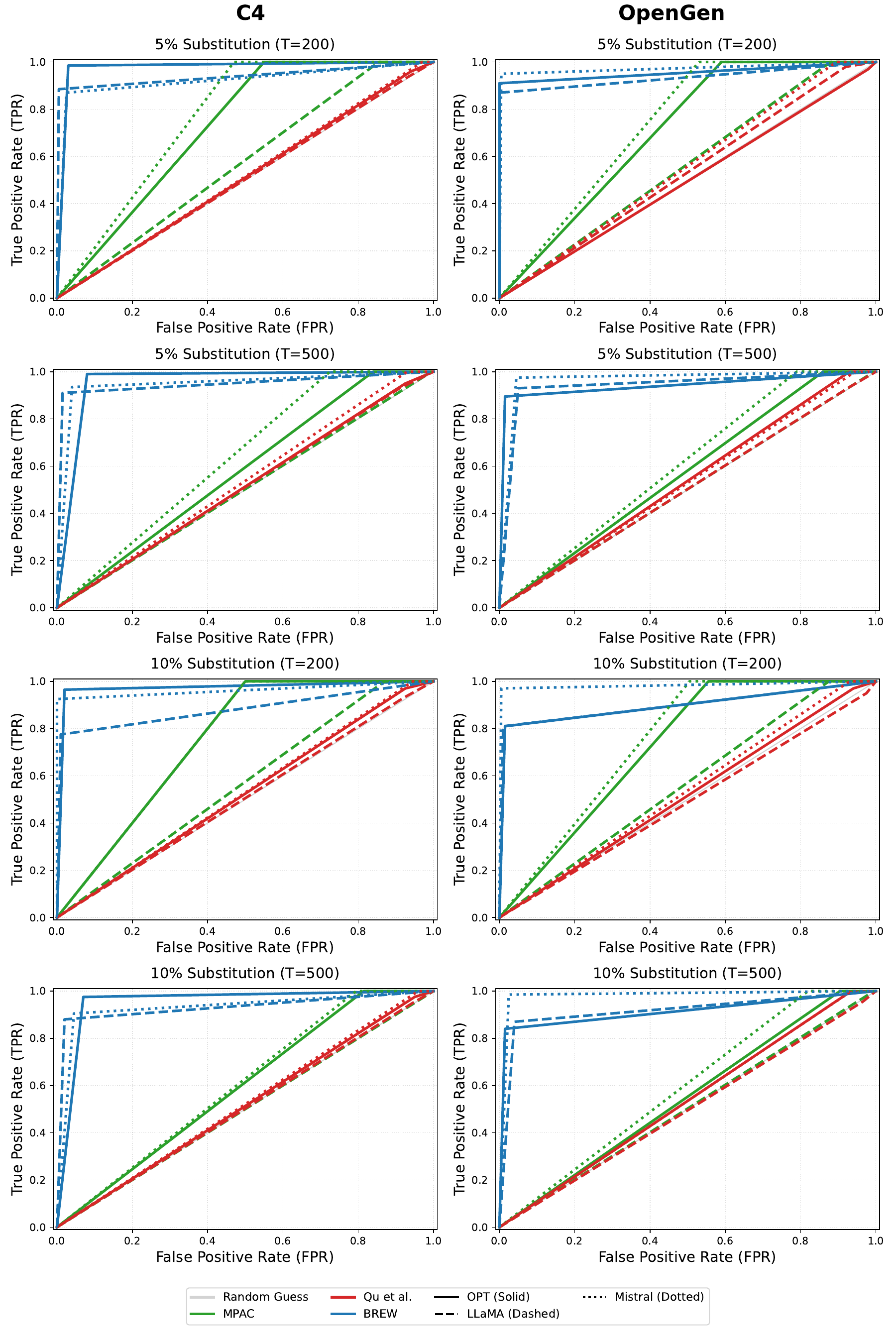}
  \caption{ROC curves under Token-preserving synonym substitution attacks across multiple backbone models. Results are shown for the C4 (left) and OpenGen (right) datasets. Rows correspond to substitution rates and text lengths: 5\% ($T{=}200$), 5\% ($T{=}500$), 10\% ($T{=}200$), and 10\% ($T{=}500$) from top to bottom. Colors denote watermarking methods (BREW, MPAC, Qu et al., and random guessing), while line styles distinguish backbone models (OPT-1.3B: solid, LLaMA-3.2-3B: dashed, Mistral-7B: dotted). Consistent detection trends are observed across all models and datasets.}
  \label{fig:Substitution_ROC_Final}
\end{figure}

\begin{figure}[t]
  \centering
  \includegraphics[
    width=\columnwidth,
    height=0.9\textheight,
    keepaspectratio
  ]{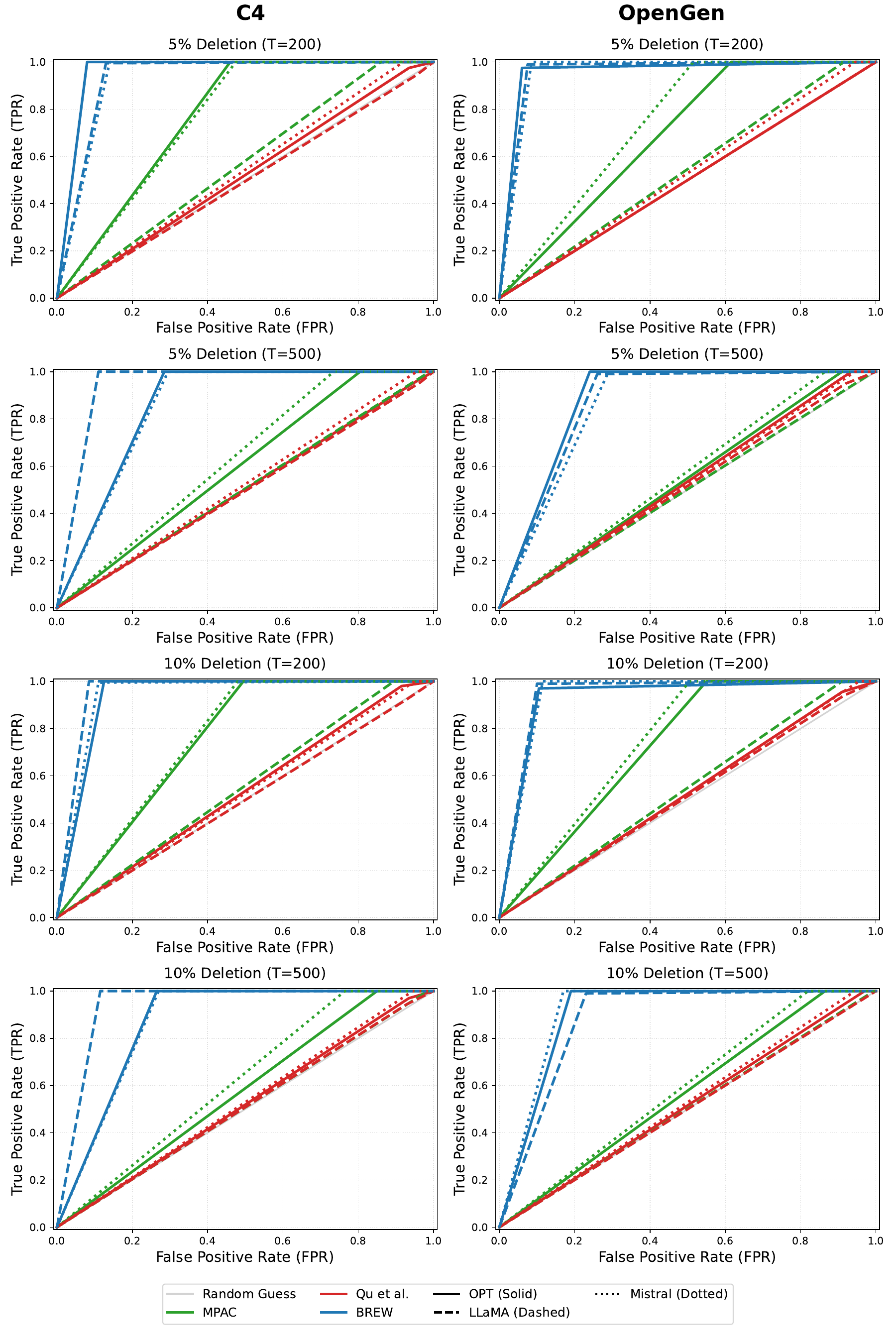}
  \caption{ROC curves under Token-reducing synonym substitution attacks across multiple backbone models. Results are shown for the C4 (left) and OpenGen (right) datasets. Rows correspond to substitution rates and text lengths: 5\% ($T{=}200$), 5\% ($T{=}500$), 10\% ($T{=}200$), and 10\% ($T{=}500$) from top to bottom. Colors denote watermarking methods (BREW, MPAC, Qu et al., and random guessing), while line styles distinguish backbone models (OPT-1.3B: solid, LLaMA-3.2-3B: dashed, Mistral-7B: dotted). Consistent detection trends are observed across all models and datasets.}
  \label{fig:Deletion_ROC_Final}
\end{figure}

\begin{figure}[t]
  \centering
  \includegraphics[
    width=\columnwidth,
    height=0.9\textheight,
    keepaspectratio
  ]{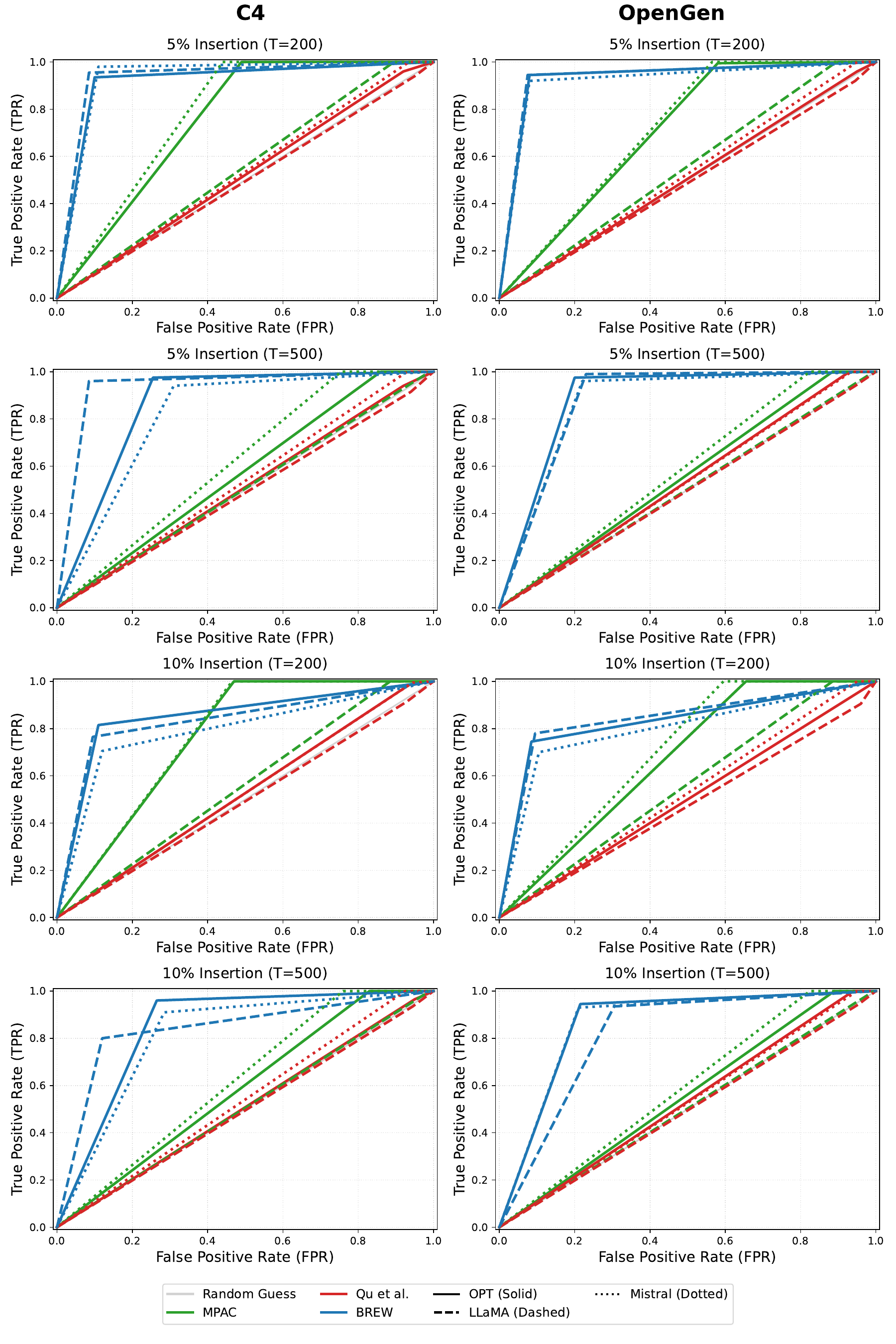}
  \caption{ROC curves under Token-increasing synonym substitution attacks across multiple backbone models. Results are shown for the C4 (left) and OpenGen (right) datasets. Rows correspond to substitution rates and text lengths: 5\% ($T{=}200$), 5\% ($T{=}500$), 10\% ($T{=}200$), and 10\% ($T{=}500$) from top to bottom. Colors denote watermarking methods (BREW, MPAC, Qu et al., and random guessing), while line styles distinguish backbone models (OPT-1.3B: solid, LLaMA-3.2-3B: dashed, Mistral-7B: dotted). Consistent detection trends are observed across all models and datasets.}
  \label{fig:Insertion_ROC_Final}
\end{figure}

This section reports comprehensive ROC results for synonym substitution attacks across backbone models, datasets, substitution rates, and text lengths. We distinguish three token-level effects: token-preserving substitutions (Figure~\ref{fig:Substitution_ROC_Final}), token-reducing (deletion-like) substitutions (Figure~\ref{fig:Deletion_ROC_Final}), and token-increasing (insertion-like) substitutions (Figure~\ref{fig:Insertion_ROC_Final}). Across all settings, BREW consistently separates from the random-guess diagonal, achieving strong TPR in the low-FPR regime across OPT-1.3B, LLaMA-3.2-3B, and Mistral-7B on both C4 and OpenGen, for 5\%/10\% substitution and $T\in\{200,500\}$.

Under \textbf{deletion-like} substitutions on C4, BREW attains near-perfect TPR with moderate FPR; e.g., at 10\% deletion on OPT-1.3B, $\mathrm{TPR}\approx0.91$--$1.00$ with $\mathrm{FPR}\approx0.075$--$0.125$ for $T{=}200$, and $\mathrm{TPR}\approx0.995$--$1.00$ with $\mathrm{FPR}\approx0.19$--$0.32$ for $T{=}500$ (depending on $\delta$). Similar trends hold for LLaMA and Mistral, with high TPR (typically $\ge0.95$ for $\delta\ge2$) while maintaining substantially lower FPR than MPAC~\cite{yoo2024advancing} and \cite{qu2025provably}. In contrast, \cite{qu2025provably} exhibits pervasive false positives, with FPR typically around $0.9$--$0.97$ across C4/OpenGen and both $T{=}200$ and $T{=}500$, despite high TPR. MPAC~\cite{yoo2024advancing} shows intermediate behavior: TPR often saturates near $1.0$, but FPR is substantially elevated and increases with text length (e.g., on C4 under deletion-like attacks, $\mathrm{FPR}\approx0.46$--$0.55$ for $T{=}200$ and $\mathrm{FPR}\approx0.73$--$0.87$ for $T{=}500$).

For \textbf{insertion-like} substitutions, BREW’s advantage is most pronounced in the low-FPR regime: at 10\% insertion on OpenGen with OPT-1.3B, BREW achieves $\mathrm{FPR}\approx0.08$--$0.09$ while improving TPR from $\approx0.41$ to $\approx0.75$ at $T{=}200$, and from $\approx0.55$ to $\approx0.95$ at $T{=}500$ as $\delta$ increases. Overall, the method ordering is consistent across token-preserving and token-altering regimes, confirming that BREW’s robustness generalizes across diverse token-level perturbations.

\subsubsection{Full Results on OPT-1.3B}
\label{app:opt_full_results}

Tables~\ref{tab:opt13b_c4_synonym_substitution_5}--\ref{tab:opt13b_opengen_synonym_insertion_10} report full detection results on OPT-1.3B.


\begin{table}[h!]
\centering
\caption{Detection performance under 5\% token-preserving synonym substitution using OPT-1.3B on the C4 dataset.}
\label{tab:opt13b_c4_synonym_substitution_5}
\renewcommand{\arraystretch}{1.2}
\begin{tabular}{l c cccc cccc}
\toprule
\multicolumn{2}{c|}{\textbf{Setting}} & \multicolumn{4}{c|}{\textbf{T200}} & \multicolumn{4}{c}{\textbf{T500}} \\
\cmidrule(lr){1-2}\cmidrule(lr){3-6}\cmidrule(lr){7-10}
\textbf{Model} & $\delta$ & TPR & FPR & Precision & F1 & TPR & FPR & Precision & F1 \\
\midrule
MPAC
 & 1.5 & 0.9950 & 0.5550 & 0.6419 & 0.7804 & 1.0000 & 0.8200 & 0.5495 & 0.7092 \\
 & 2.0 & 1.0000 & 0.5400 & 0.6494 & 0.7874 & 1.0000 & 0.8650 & 0.5362 & 0.6981 \\
 & 3.0 & 1.0000 & 0.5300 & 0.6536 & 0.7905 & 1.0000 & 0.8550 & 0.5391 & 0.7005 \\
 & 6.0 & 1.0000 & 0.5500 & 0.6452 & 0.7842 & 1.0000 & 0.8400 & 0.5435 & 0.7042 \\
\midrule
\cite{qu2025provably}
 & 1.5 & 0.9600 & 0.9350 & 0.5066 & 0.6632 & 1.0000 & 0.9100 & 0.5236 & 0.6873 \\
 & 2.0 & 0.9550 & 0.9350 & 0.5053 & 0.6609 & 1.0000 & 0.9300 & 0.5181 & 0.6826 \\
 & 3.0 & 0.9800 & 0.9300 & 0.5131 & 0.6735 & 0.9650 & 0.9600 & 0.5013 & 0.6598 \\
 & 6.0 & 0.9600 & 0.9400 & 0.5053 & 0.6621 & 0.9500 & 0.9250 & 0.5067 & 0.6609 \\
\midrule
BREW (Ours)
 & 1.5 & 0.6450 & 0.0150 & 0.9773 & 0.7771 & 0.7400 & 0.0300 & 0.9610 & 0.8362 \\
 & 2.0 & 0.7600 & 0.0100 & 0.9870 & 0.8588 & 0.7450 & 0.0250 & 0.9675 & 0.8418 \\
 & 3.0 & 0.8050 & 0.0000 & 1.0000 & 0.8920 & 0.9000 & 0.0250 & 0.9730 & 0.9351 \\
 & 6.0 & 0.9850 & 0.0300 & 0.9704 & 0.9777 & 0.9900 & 0.0800 & 0.9252 & 0.9565 \\
\bottomrule
\end{tabular}
\end{table}

\begin{table}[h!]
\centering
\caption{Detection performance under 10\% token-preserving synonym substitution using OPT-1.3B on the C4 dataset.}
\label{tab:opt13b_c4_synonym_substitution_10}
\renewcommand{\arraystretch}{1.2}
\begin{tabular}{l c cccc cccc}
\toprule
\multicolumn{2}{c|}{\textbf{Setting}} & \multicolumn{4}{c|}{\textbf{T200}} & \multicolumn{4}{c}{\textbf{T500}} \\
\cmidrule(lr){1-2}\cmidrule(lr){3-6}\cmidrule(lr){7-10}
\textbf{Model} & $\delta$ & TPR & FPR & Precision & F1 & TPR & FPR & Precision & F1 \\
\midrule
MPAC
 & 1.5 & 0.9900 & 0.5750 & 0.6326 & 0.7719 & 1.0000 & 0.8300 & 0.5464 & 0.7067 \\
 & 2.0 & 1.0000 & 0.5500 & 0.6452 & 0.7843 & 1.0000 & 0.8100 & 0.5525 & 0.7117 \\
 & 3.0 & 1.0000 & 0.5100 & 0.6623 & 0.7968 & 1.0000 & 0.8400 & 0.5435 & 0.7042 \\
 & 6.0 & 1.0000 & 0.5000 & 0.6667 & 0.8000 & 1.0000 & 0.8150 & 0.5510 & 0.7105 \\
\midrule
\cite{qu2025provably}
 & 1.5 & 0.9650 & 0.9250 & 0.5106 & 0.6678 & 0.9900 & 0.9000 & 0.5238 & 0.6851 \\
 & 2.0 & 0.9600 & 0.9250 & 0.5093 & 0.6655 & 1.0000 & 0.9400 & 0.5155 & 0.6803 \\
 & 3.0 & 0.9800 & 0.9200 & 0.5158 & 0.6759 & 0.9650 & 0.9400 & 0.5066 & 0.6644 \\
 & 6.0 & 0.9700 & 0.9250 & 0.5132 & 0.6724 & 0.9750 & 0.9500 & 0.5065 & 0.6667 \\
\midrule
BREW (Ours)
 & 1.5 & 0.5850 & 0.0100 & 0.9832 & 0.7335 & 0.4700 & 0.0300 & 0.9400 & 0.6267 \\
 & 2.0 & 0.6300 & 0.0050 & 0.9921 & 0.7706 & 0.6750 & 0.0050 & 0.9926 & 0.8036 \\
 & 3.0 & 0.7600 & 0.0000 & 1.0000 & 0.8636 & 0.7650 & 0.0400 & 0.9503 & 0.8476 \\
 & 6.0 & 0.9650 & 0.0200 & 0.9797 & 0.9723 & 0.9750 & 0.0700 & 0.9330 & 0.9535 \\
\bottomrule
\end{tabular}
\end{table}


\begin{table}[h!]
\centering
\caption{Detection performance under 5\% token-preserving synonym substitution using OPT-1.3B on the OpenGen dataset.}
\label{tab:opt13b_opengen_synonym_substitution_5}
\renewcommand{\arraystretch}{1.2}
\begin{tabular}{l c cccc cccc}
\toprule
\multicolumn{2}{c|}{\textbf{Setting}} & \multicolumn{4}{c|}{\textbf{T200}} & \multicolumn{4}{c}{\textbf{T500}} \\
\cmidrule(lr){1-2}\cmidrule(lr){3-6}\cmidrule(lr){7-10}
\textbf{Model} & $\delta$ & TPR & FPR & Precision & F1 & TPR & FPR & Precision & F1 \\
\midrule
MPAC
 & 1.5 & 0.9850 & 0.6350 & 0.6080 & 0.7519 & 1.0000 & 0.8600 & 0.5376 & 0.6993 \\
 & 2.0 & 1.0000 & 0.5750 & 0.6349 & 0.7767 & 1.0000 & 0.8800 & 0.5319 & 0.6944 \\
 & 3.0 & 1.0000 & 0.6450 & 0.6079 & 0.7561 & 1.0000 & 0.9000 & 0.5263 & 0.6897 \\
 & 6.0 & 1.0000 & 0.5900 & 0.6289 & 0.7722 & 1.0000 & 0.8600 & 0.5376 & 0.6993 \\
\midrule
\cite{qu2025provably} 
 & 1.5 & 0.9550 & 0.9500 & 0.5013 & 0.6574 & 0.9850 & 0.9200 & 0.5171 & 0.6781 \\
 & 2.0 & 0.9700 & 0.9200 & 0.5132 & 0.6712 & 1.0000 & 0.9650 & 0.5089 & 0.6745 \\
 & 3.0 & 0.9900 & 0.9100 & 0.5211 & 0.6827 & 1.0000 & 0.9300 & 0.5181 & 0.6825 \\
 & 6.0 & 0.9700 & 0.9800 & 0.4974 & 0.6576 & 1.0000 & 0.9300 & 0.5181 & 0.6825 \\
\midrule
BREW (Ours)
 & 1.5 & 0.5850 & 0.0050 & 0.9915 & 0.7358 & 0.6150 & 0.0400 & 0.9389 & 0.7432 \\
 & 2.0 & 0.8000 & 0.0050 & 0.9937 & 0.8864 & 0.8100 & 0.0350 & 0.9585 & 0.8780 \\
 & 3.0 & 0.8650 & 0.0000 & 1.0000 & 0.9276 & 0.8700 & 0.0150 & 0.9831 & 0.9231 \\
 & 6.0 & 0.9100 & 0.0000 & 1.0000 & 0.9528 & 0.8950 & 0.0150 & 0.9835 & 0.9371 \\
\bottomrule
\end{tabular}
\end{table}

\begin{table}[h!]
\centering
\caption{Detection performance under 10\% token-preserving synonym substitution using OPT-1.3B on the OpenGen dataset.}
\label{tab:opt13b_opengen_synonym_substitution_10}
\renewcommand{\arraystretch}{1.2}
\begin{tabular}{l c cccc cccc}
\toprule
\multicolumn{2}{c|}{\textbf{Setting}} & \multicolumn{4}{c|}{\textbf{T200}} & \multicolumn{4}{c}{\textbf{T500}} \\
\cmidrule(lr){1-2}\cmidrule(lr){3-6}\cmidrule(lr){7-10}
\textbf{Model} & $\delta$ & TPR & FPR & Precision & F1 & TPR & FPR & Precision & F1 \\
\midrule
MPAC
 & 1.5 & 0.9900 & 0.6150 & 0.6168 & 0.7601 & 1.0000 & 0.9300 & 0.5181 & 0.6826 \\
 & 2.0 & 1.0000 & 0.6300 & 0.6135 & 0.7605 & 1.0000 & 0.8750 & 0.5333 & 0.6957 \\
 & 3.0 & 1.0000 & 0.5250 & 0.6557 & 0.7921 & 1.0000 & 0.8800 & 0.5319 & 0.6944 \\
 & 6.0 & 1.0000 & 0.5550 & 0.6431 & 0.7828 & 1.0000 & 0.9050 & 0.5249 & 0.6885 \\
\midrule
\cite{qu2025provably} 
 & 1.5 & 0.9500 & 0.9250 & 0.5066 & 0.6608 & 0.9850 & 0.9600 & 0.5064 & 0.6689 \\
 & 2.0 & 0.9550 & 0.9200 & 0.5093 & 0.6643 & 0.9950 & 0.9250 & 0.5182 & 0.6815 \\
 & 3.0 & 0.9850 & 0.9450 & 0.5103 & 0.6723 & 1.0000 & 0.9400 & 0.5154 & 0.6802 \\
 & 6.0 & 0.9700 & 0.9400 & 0.5078 & 0.6666 & 0.9950 & 0.9300 & 0.5168 & 0.6803 \\
\midrule
BREW (Ours)
 & 1.5 & 0.4300 & 0.0100 & 0.9772 & 0.5972 & 0.4300 & 0.0200 & 0.9555 & 0.5931 \\
 & 2.0 & 0.6800 & 0.0250 & 0.9645 & 0.7976 & 0.7150 & 0.0350 & 0.9533 & 0.8171 \\
 & 3.0 & 0.7850 & 0.0050 & 0.9936 & 0.8771 & 0.7650 & 0.0150 & 0.9807 & 0.8595 \\
 & 6.0 & 0.8100 & 0.0150 & 0.9818 & 0.8876 & 0.8400 & 0.0150 & 0.9824 & 0.9056 \\
\bottomrule
\end{tabular}
\end{table}


\begin{table}[t]
\centering
\caption{Detection performance under 5\% token-reducing synonym substitution using OPT-1.3B on the C4 dataset.}
\label{tab:opt13b_c4_synonym_deletion_5}
\renewcommand{\arraystretch}{1.2}
\begin{tabular}{l l l cccc cccc}
\toprule
\multicolumn{3}{c|}{\textbf{Setting}} & \multicolumn{4}{c|}{\textbf{T200}} & \multicolumn{4}{c}{\textbf{T500}} \\
\cmidrule(lr){1-3}\cmidrule(lr){4-7}\cmidrule(lr){8-11}
\textbf{Model} & $\delta$ & $s_{\max}$ & TPR & FPR & Precision & F1 & TPR & FPR & Precision & F1 \\
\midrule
MPAC 
 & 1.5 & - & 0.9950 & 0.5750 & 0.6338 & 0.7743 & 1.0000 & 0.8000 & 0.5556 & 0.7143 \\
 & 2.0 & - & 1.0000 & 0.4850 & 0.6734 & 0.8048 & 1.0000 & 0.8650 & 0.5362 & 0.6981 \\
 & 3.0 & - & 1.0000 & 0.5300 & 0.6536 & 0.7905 & 1.0000 & 0.8650 & 0.5362 & 0.6981 \\
 & 6.0 & - & 1.0000 & 0.4600 & 0.6849 & 0.8130 & 1.0000 & 0.8050 & 0.5540 & 0.7130 \\
\midrule
\cite{qu2025provably}
 & 1.5 & - & 0.9400 & 0.9800 & 0.4896 & 0.6438 & 0.9800 & 0.9350 & 0.5117 & 0.6724 \\
 & 2.0 & - & 0.9700 & 0.9550 & 0.5039 & 0.6632 & 0.9750 & 0.9350 & 0.5105 & 0.6701 \\
 & 3.0 & - & 0.9950 & 0.9400 & 0.5142 & 0.6780 & 0.9700 & 0.9200 & 0.5132 & 0.6713 \\
 & 6.0 & - & 0.9750 & 0.9350 & 0.5105 & 0.6701 & 0.9650 & 0.9600 & 0.5013 & 0.6598 \\
\midrule
BREW (Ours) & 1.5
 & 0 & 0.8100 & 0.0000 & 1.0000 & 0.8950 & 0.9000 & 0.0300 & 0.9677 & 0.9326 \\
 & & 1 & 0.8950 & 0.0250 & 0.9728 & 0.9323 & 0.9850 & 0.0750 & 0.9292 & 0.9563 \\
 & & 3 & 0.9000 & 0.0950 & 0.9045 & 0.9023 & 0.9950 & 0.1400 & 0.8767 & 0.9321 \\
 & & 5 & 0.9150 & 0.1200 & 0.8841 & 0.8993 & 1.0000 & 0.2450 & 0.8032 & 0.8909 \\
\cmidrule(lr){2-11}
 & 2.0
 & 0 & 0.9050 & 0.0050 & 0.9945 & 0.9476 & 0.9050 & 0.0150 & 0.9837 & 0.9427 \\
 & & 1 & 0.9850 & 0.0250 & 0.9752 & 0.9801 & 0.9950 & 0.0600 & 0.9431 & 0.9684 \\
 & & 3 & 0.9950 & 0.0450 & 0.9567 & 0.9755 & 1.0000 & 0.1300 & 0.8850 & 0.9390 \\
 & & 5 & 0.9950 & 0.0800 & 0.9256 & 0.9590 & 1.0000 & 0.2500 & 0.8000 & 0.8888 \\
\cmidrule(lr){2-11}
 & 3.0
 & 0 & 0.9200 & 0.0050 & 0.9946 & 0.9558 & 0.9050 & 0.0200 & 0.9784 & 0.9403 \\
 & & 1 & 0.9950 & 0.0300 & 0.9707 & 0.9827 & 0.9950 & 0.0750 & 0.9299 & 0.9614 \\
 & & 3 & 1.0000 & 0.0550 & 0.9479 & 0.9732 & 1.0000 & 0.1850 & 0.8439 & 0.9153 \\
 & & 5 & 1.0000 & 0.0900 & 0.9174 & 0.9569 & 1.0000 & 0.2600 & 0.7937 & 0.8850 \\
\cmidrule(lr){2-11}
 & 6.0
 & 0 & 0.9200 & 0.0000 & 1.0000 & 0.9583 & 0.9500 & 0.0300 & 0.9694 & 0.9596 \\
 & & 1 & 1.0000 & 0.0450 & 0.9569 & 0.9780 & 0.9950 & 0.0750 & 0.9299 & 0.9614 \\
 & & 3 & 1.0000 & 0.0700 & 0.9346 & 0.9662 & 1.0000 & 0.2000 & 0.8333 & 0.9091 \\
 & & 5 & 1.0000 & 0.0800 & 0.9259 & 0.9615 & 1.0000 & 0.2850 & 0.7782 & 0.8753 \\
\bottomrule
\end{tabular}
\end{table}

\begin{table}[t]
\centering
\caption{Detection performance under 10\% token-reducing synonym substitution using OPT-1.3B on the C4 dataset.}
\label{tab:opt13b_c4_synonym_deletion_10}
\renewcommand{\arraystretch}{1.2}
\begin{tabular}{l l l cccc cccc}
\toprule
\multicolumn{3}{c|}{\textbf{Setting}} & \multicolumn{4}{c|}{\textbf{T200}} & \multicolumn{4}{c}{\textbf{T500}} \\
\cmidrule(lr){1-3}\cmidrule(lr){4-7}\cmidrule(lr){8-11}
\textbf{Model} & $\delta$ & $s_{\max}$ & TPR & FPR & Precision & F1 & TPR & FPR & Precision & F1 \\
\midrule
MPAC
 & 1.5 & - & 0.9950 & 0.4800 & 0.6746 & 0.8040 & 1.0000 & 0.8050 & 0.5540 & 0.7130 \\
 & 2.0 & - & 1.0000 & 0.5500 & 0.6452 & 0.7843 & 1.0000 & 0.7900 & 0.5587 & 0.7168 \\
 & 3.0 & - & 1.0000 & 0.5450 & 0.6472 & 0.7859 & 1.0000 & 0.8250 & 0.5479 & 0.7080 \\
 & 6.0 & - & 1.0000 & 0.4950 & 0.6689 & 0.8016 & 1.0000 & 0.8500 & 0.5405 & 0.7018 \\
\midrule
\cite{qu2025provably}
 & 1.5 & - & 0.9400 & 0.9250 & 0.5040 & 0.6562 & 0.9900 & 0.9500 & 0.5103 & 0.6735 \\
 & 2.0 & - & 0.9650 & 0.9250 & 0.5106 & 0.6678 & 0.9700 & 0.9400 & 0.5079 & 0.6667 \\
 & 3.0 & - & 0.9850 & 0.9100 & 0.5198 & 0.6805 & 0.9550 & 0.9450 & 0.5026 & 0.6586 \\
 & 6.0 & - & 0.9800 & 0.9150 & 0.5172 & 0.6770 & 0.9700 & 0.9350 & 0.5092 & 0.6678 \\
\midrule
BREW (Ours) & 1.5
 & 0 & 0.7500 & 0.0050 & 0.9934 & 0.8547 & 0.8700 & 0.0300 & 0.9667 & 0.8865 \\
 & & 1 & 0.9000 & 0.0200 & 0.9783 & 0.9375 & 0.9650 & 0.0650 & 0.9369 & 0.9508 \\
 & & 3 & 0.9050 & 0.0550 & 0.9427 & 0.9235 & 0.9800 & 0.1750 & 0.8485 & 0.9095 \\
 & & 5 & 0.9100 & 0.1000 & 0.9010 & 0.9055 & 0.9950 & 0.3200 & 0.7567 & 0.8596 \\
\cmidrule(lr){2-11}
 & 2.0
 & 0 & 0.8750 & 0.0200 & 0.9777 & 0.9235 & 0.9000 & 0.0300 & 0.9611 & 0.9326 \\
 & & 1 & 0.9850 & 0.0300 & 0.9704 & 0.9777 & 0.9850 & 0.0400 & 0.9610 & 0.9728 \\
 & & 3 & 0.9950 & 0.0800 & 0.9256 & 0.9590 & 0.9950 & 0.1500 & 0.8690 & 0.9277 \\
 & & 5 & 0.9950 & 0.0750 & 0.9299 & 0.9614 & 1.0000 & 0.1900 & 0.8403 & 0.9132 \\
\cmidrule(lr){2-11}
 & 3.0
 & 0 & 0.8950 & 0.0100 & 0.9890 & 0.9396 & 0.9350 & 0.0200 & 0.9791 & 0.9565 \\
 & & 1 & 1.0000 & 0.0100 & 0.9901 & 0.9950 & 0.9950 & 0.0750 & 0.9299 & 0.9614 \\
 & & 3 & 1.0000 & 0.0450 & 0.9569 & 0.9780 & 1.0000 & 0.1600 & 0.8621 & 0.9259 \\
 & & 5 & 1.0000 & 0.1050 & 0.9050 & 0.9501 & 1.0000 & 0.2800 & 0.7813 & 0.8772 \\
\cmidrule(lr){2-11}
 & 6.0
 & 0 & 0.9050 & 0.0200 & 0.9784 & 0.9403 & 0.9250 & 0.0300 & 0.9686 & 0.9463 \\
 & & 1 & 0.9900 & 0.0250 & 0.9754 & 0.9826 & 0.9950 & 0.0800 & 0.9256 & 0.9590 \\
 & & 3 & 1.0000 & 0.0700 & 0.9346 & 0.9662 & 1.0000 & 0.1650 & 0.8584 & 0.9238 \\
 & & 5 & 1.0000 & 0.1250 & 0.8889 & 0.9412 & 1.0000 & 0.2650 & 0.7905 & 0.8830 \\
\bottomrule
\end{tabular}
\end{table}

\begin{table}[t]
\centering
\caption{Detection performance under 5\% token-reducing synonym substitution using OPT-1.3B on the OpenGen dataset.}
\label{tab:opt13b_opengen_synonym_deletion_5}
\renewcommand{\arraystretch}{1.2}
\begin{tabular}{l l l cccc cccc}
\toprule
\multicolumn{3}{c|}{\textbf{Setting}} & \multicolumn{4}{c|}{\textbf{T200}} & \multicolumn{4}{c}{\textbf{T500}} \\
\cmidrule(lr){1-3}\cmidrule(lr){4-7}\cmidrule(lr){8-11}
\textbf{Model} & $\delta$ & $s_{\max}$ & TPR & FPR & Precision & F1 & TPR & FPR & Precision & F1 \\
\midrule
MPAC
 & 1.5 & - & 0.9950 & 0.5650 & 0.6378 & 0.7773 & 1.0000 & 0.9150 & 0.5222 & 0.6861 \\
 & 2.0 & - & 1.0000 & 0.5200 & 0.6579 & 0.7937 & 1.0000 & 0.8450 & 0.5420 & 0.7030 \\
 & 3.0 & - & 1.0000 & 0.5800 & 0.6329 & 0.7752 & 1.0000 & 0.8700 & 0.5348 & 0.6969 \\
 & 6.0 & - & 1.0000 & 0.6150 & 0.6192 & 0.7648 & 1.0000 & 0.9100 & 0.5236 & 0.6873 \\
\midrule
\cite{qu2025provably}
 & 1.5 & - & 0.9550 & 0.9050 & 0.5134 & 0.6678 & 0.9950 & 0.9550 & 0.5102 & 0.6745 \\
 & 2.0 & - & 0.9600 & 0.9200 & 0.5106 & 0.6666 & 0.9950 & 0.9400 & 0.5142 & 0.6780 \\
 & 3.0 & - & 1.0000 & 0.9600 & 0.5102 & 0.6756 & 1.0000 & 0.9700 & 0.5076 & 0.6734 \\
 & 6.0 & - & 0.9550 & 0.9550 & 0.5000 & 0.6563 & 1.0000 & 0.9350 & 0.5167 & 0.6814 \\
\midrule
BREW (Ours) & 1.5
 & 0 & 0.8450 & 0.0150 & 0.9825 & 0.9086 & 0.8850 & 0.0250 & 0.9725 & 0.9267 \\
 & & 1 & 0.8800 & 0.0200 & 0.9777 & 0.9263 & 0.9350 & 0.0800 & 0.9211 & 0.9280 \\
 & & 3 & 0.8550 & 0.0850 & 0.9095 & 0.8814 & 1.0000 & 0.1200 & 0.8929 & 0.9434 \\
 & & 5 & 0.8950 & 0.1000 & 0.8994 & 0.8972 & 0.9850 & 0.2200 & 0.8174 & 0.8934 \\
\cmidrule(lr){2-11}
 & 2.0
 & 0 & 0.9300 & 0.0000 & 1.0000 & 0.9637 & 0.9350 & 0.0150 & 0.9842 & 0.9590 \\
 & & 1 & 0.9400 & 0.0350 & 0.9641 & 0.9518 & 1.0000 & 0.0900 & 0.9174 & 0.9569 \\
 & & 3 & 0.9600 & 0.0700 & 0.9320 & 0.9458 & 1.0000 & 0.1300 & 0.8850 & 0.9390 \\
 & & 5 & 0.9400 & 0.0800 & 0.9215 & 0.9306 & 1.0000 & 0.1950 & 0.8368 & 0.9116 \\
\cmidrule(lr){2-11}
 & 3.0
 & 0 & 0.9450 & 0.0000 & 1.0000 & 0.9717 & 0.9450 & 0.0250 & 0.9742 & 0.9594 \\
 & & 1 & 0.9650 & 0.0150 & 0.9846 & 0.9747 & 1.0000 & 0.0400 & 0.9615 & 0.9803 \\
 & & 3 & 0.9550 & 0.0650 & 0.9362 & 0.9455 & 1.0000 & 0.1800 & 0.8475 & 0.9174 \\
 & & 5 & 0.9450 & 0.0300 & 0.9692 & 0.9569 & 1.0000 & 0.2350 & 0.8097 & 0.8948 \\
\cmidrule(lr){2-11}
 & 6.0
 & 0 & 0.9400 & 0.0050 & 0.9947 & 0.9665 & 0.9650 & 0.0300 & 0.9698 & 0.9674 \\
 & & 1 & 0.9550 & 0.0250 & 0.9744 & 0.9646 & 0.9900 & 0.0850 & 0.9209 & 0.9542 \\
 & & 3 & 0.9950 & 0.0400 & 0.9613 & 0.9778 & 0.7950 & 0.0200 & 0.9755 & 0.8760 \\
 & & 5 & 0.9750 & 0.0600 & 0.9420 & 0.9582 & 1.0000 & 0.2400 & 0.8064 & 0.8928 \\
\bottomrule
\end{tabular}
\end{table}

\begin{table}[t]
\centering
\caption{Detection performance under 10\% token-reducing synonym substitution using OPT-1.3B on the OpenGen dataset.}
\label{tab:opt13b_opengen_synonym_deletion_10}
\renewcommand{\arraystretch}{1.2}
\begin{tabular}{l l l cccc cccc}
\toprule
\multicolumn{3}{c|}{\textbf{Setting}} & \multicolumn{4}{c|}{\textbf{T200}} & \multicolumn{4}{c}{\textbf{T500}} \\
\cmidrule(lr){1-3}\cmidrule(lr){4-7}\cmidrule(lr){8-11}
\textbf{Model} & $\delta$ & $s_{\max}$ & TPR & FPR & Precision & F1 & TPR & FPR & Precision & F1 \\
\midrule
MPAC
 & 1.5 & - & 0.9900 & 0.5250 & 0.6535 & 0.7873 & 1.0000 & 0.8750 & 0.5333 & 0.6957 \\
 & 2.0 & - & 1.0000 & 0.5550 & 0.6431 & 0.7828 & 1.0000 & 0.8600 & 0.5376 & 0.6993 \\
 & 3.0 & - & 1.0000 & 0.5400 & 0.6494 & 0.7874 & 1.0000 & 0.8250 & 0.5479 & 0.7080 \\
 & 6.0 & - & 1.0000 & 0.5500 & 0.6452 & 0.7843 & 1.0000 & 0.8650 & 0.5362 & 0.6981 \\
\midrule
\cite{qu2025provably}
 & 1.5 & - & 0.9500 & 0.9400 & 0.5026 & 0.6574 & 0.9750 & 0.9600 & 0.5039 & 0.6644 \\
 & 2.0 & - & 0.9700 & 0.9650 & 0.5012 & 0.6609 & 1.0000 & 0.9100 & 0.5235 & 0.6872 \\
 & 3.0 & - & 0.9950 & 0.9400 & 0.5142 & 0.6780 & 1.0000 & 0.9450 & 0.5141 & 0.6791 \\
 & 6.0 & - & 0.9550 & 0.9100 & 0.5121 & 0.6666 & 1.0000 & 0.9700 & 0.5076 & 0.6734 \\
\midrule
BREW (Ours) & 1.5
 & 0 & 0.7650 & 0.0050 & 0.9935 & 0.8644 & 0.8600 & 0.0250 & 0.9718 & 0.9125 \\
 & & 1 & 0.7950 & 0.0200 & 0.9754 & 0.8760 & 0.9250 & 0.0400 & 0.9585 & 0.9414 \\
 & & 3 & 0.8650 & 0.0650 & 0.9301 & 0.8963 & 0.8300 & 0.0300 & 0.9651 & 0.8924 \\
 & & 5 & 0.8650 & 0.0650 & 0.9301 & 0.8963 & 0.9650 & 0.2100 & 0.8212 & 0.8873 \\
\cmidrule(lr){2-11}
 & 2.0
 & 0 & 0.8800 & 0.0100 & 0.9888 & 0.9312 & 0.8950 & 0.0200 & 0.9781 & 0.9347 \\
 & & 1 & 0.9500 & 0.0350 & 0.9644 & 0.9571 & 0.9950 & 0.0700 & 0.9342 & 0.9636 \\
 & & 3 & 0.9650 & 0.0550 & 0.9460 & 0.9554 & 1.0000 & 0.1600 & 0.8621 & 0.9259 \\
 & & 5 & 0.9350 & 0.0900 & 0.9121 & 0.9234 & 0.9950 & 0.2100 & 0.8257 & 0.9024 \\
\cmidrule(lr){2-11}
 & 3.0
 & 0 & 0.9250 & 0.0050 & 0.9946 & 0.9585 & 0.9400 & 0.0250 & 0.9741 & 0.9567 \\
 & & 1 & 0.9500 & 0.0100 & 0.9895 & 0.9693 & 0.9900 & 0.0600 & 0.9428 & 0.9658 \\
 & & 3 & 0.9600 & 0.0650 & 0.9365 & 0.9481 & 1.0000 & 0.1650 & 0.8584 & 0.9238 \\
 & & 5 & 0.9850 & 0.0900 & 0.9162 & 0.9493 & 1.0000 & 0.2000 & 0.8333 & 0.9091 \\
\cmidrule(lr){2-11}
 & 6.0
 & 0 & 0.9300 & 0.0050 & 0.9946 & 0.9612 & 0.9550 & 0.0250 & 0.9744 & 0.9646 \\
 & & 1 & 0.9400 & 0.0250 & 0.9741 & 0.9567 & 1.0000 & 0.0500 & 0.9523 & 0.9756 \\
 & & 3 & 0.9850 & 0.0550 & 0.9471 & 0.9656 & 1.0000 & 0.1100 & 0.9009 & 0.9479 \\
 & & 5 & 0.9700 & 0.1050 & 0.9023 & 0.9349 & 1.0000 & 0.1900 & 0.8403 & 0.9132 \\
\bottomrule
\end{tabular}
\end{table}

\begin{table}[t]
\centering
\caption{Detection performance under 5\% token-increasing synonym substitution using OPT-1.3B on the C4 dataset.}
\label{tab:opt13b_c4_synonym_insertion_5}
\renewcommand{\arraystretch}{1.2}
\begin{tabular}{l l l cccc cccc}
\toprule
\multicolumn{3}{c|}{\textbf{Setting}} & \multicolumn{4}{c|}{\textbf{T200}} & \multicolumn{4}{c}{\textbf{T500}} \\
\cmidrule(lr){1-3}\cmidrule(lr){4-7}\cmidrule(lr){8-11}
\textbf{Model} & $\delta$ & $s_{\max}$ & TPR & FPR & Precision & F1 & TPR & FPR & Precision & F1 \\
\midrule
MPAC
 & 1.5 & - & 0.9950 & 0.5050 & 0.6633 & 0.7960 & 1.0000 & 0.8100 & 0.5525 & 0.7117 \\
 & 2.0 & - & 1.0000 & 0.5650 & 0.6390 & 0.7797 & 1.0000 & 0.8250 & 0.5479 & 0.7080 \\
 & 3.0 & - & 1.0000 & 0.5550 & 0.6431 & 0.7828 & 1.0000 & 0.8400 & 0.5435 & 0.7042 \\
 & 6.0 & - & 1.0000 & 0.4900 & 0.6711 & 0.8032 & 1.0000 & 0.8600 & 0.5376 & 0.6993 \\
\midrule
\cite{qu2025provably}
 & 1.5 & - & 0.9450 & 0.9250 & 0.5053 & 0.6585 & 0.9650 & 0.9450 & 0.5052 & 0.6632 \\
 & 2.0 & - & 0.9650 & 0.9200 & 0.5119 & 0.6690 & 0.9500 & 0.9600 & 0.4974 & 0.6529 \\
 & 3.0 & - & 0.9700 & 0.9150 & 0.5146 & 0.6724 & 0.9600 & 0.9600 & 0.5000 & 0.6575 \\
 & 6.0 & - & 0.9600 & 0.9200 & 0.5106 & 0.6667 & 0.9400 & 0.9200 & 0.5034 & 0.6573 \\
\midrule
BREW (Ours) & 1.5
 & 0 & 0.2450 & 0.0100 & 0.9601 & 0.3904 & 0.3050 & 0.0250 & 0.9242 & 0.4586 \\
 & & 1 & 0.3800 & 0.0250 & 0.9383 & 0.5409 & 0.4400 & 0.0600 & 0.8800 & 0.5867 \\
 & & 3 & 0.5250 & 0.0500 & 0.9130 & 0.6667 & 0.5000 & 0.1550 & 0.7634 & 0.6042 \\
 & & 5 & 0.6350 & 0.0750 & 0.8944 & 0.7427 & 0.7300 & 0.1950 & 0.7892 & 0.7584 \\
\cmidrule(lr){2-11}
 & 2.0
 & 0 & 0.3450 & 0.0200 & 0.9452 & 0.5055 & 0.3700 & 0.0200 & 0.9487 & 0.5324 \\
 & & 1 & 0.5150 & 0.0250 & 0.9537 & 0.6688 & 0.6350 & 0.0900 & 0.8759 & 0.7362 \\
 & & 3 & 0.6950 & 0.0900 & 0.8854 & 0.7787 & 0.7650 & 0.1300 & 0.8547 & 0.8074 \\
 & & 5 & 0.8850 & 0.1200 & 0.8806 & 0.8828 & 0.8300 & 0.2650 & 0.7580 & 0.7924 \\
\cmidrule(lr){2-11}
 & 3.0
 & 0 & 0.5200 & 0.0200 & 0.9629 & 0.6753 & 0.4750 & 0.0200 & 0.9596 & 0.6355 \\
 & & 1 & 0.7350 & 0.0500 & 0.9363 & 0.8235 & 0.6950 & 0.1150 & 0.8580 & 0.7680 \\
 & & 3 & 0.8400 & 0.0700 & 0.9231 & 0.8796 & 0.9200 & 0.1850 & 0.8326 & 0.8741 \\
 & & 5 & 0.9300 & 0.0850 & 0.9163 & 0.9231 & 0.9450 & 0.2500 & 0.7908 & 0.8610 \\
\cmidrule(lr){2-11}
 & 6.0
 & 0 & 0.6000 & 0.0200 & 0.9677 & 0.7407 & 0.5650 & 0.0200 & 0.9658 & 0.7129 \\
 & & 1 & 0.7550 & 0.0200 & 0.9742 & 0.8507 & 0.7700 & 0.0800 & 0.9059 & 0.8324 \\
 & & 3 & 0.9500 & 0.0800 & 0.9223 & 0.9360 & 0.9350 & 0.1850 & 0.8348 & 0.8821 \\
 & & 5 & 0.9350 & 0.1000 & 0.9034 & 0.9189 & 0.9750 & 0.2550 & 0.7927 & 0.8744 \\
\bottomrule
\end{tabular}
\end{table}

\begin{table}[t]
\centering
\caption{Detection performance under 10\% token-increasing synonym substitution using OPT-1.3B on the C4 dataset.}
\label{tab:opt13b_c4_synonym_insertion_10}
\renewcommand{\arraystretch}{1.2}
\begin{tabular}{l l l cccc cccc}
\toprule
\multicolumn{3}{c|}{\textbf{Setting}} & \multicolumn{4}{c|}{\textbf{T200}} & \multicolumn{4}{c}{\textbf{T500}} \\
\cmidrule(lr){1-3}\cmidrule(lr){4-7}\cmidrule(lr){8-11}
\textbf{Model} & $\delta$ & $s_{\max}$ & TPR & FPR & Precision & F1 & TPR & FPR & Precision & F1 \\
\midrule
MPAC 
 & 1.5 & - & 0.9850 & 0.5400 & 0.6459 & 0.7802 & 1.0000 & 0.8650 & 0.5362 & 0.6981 \\
 & 2.0 & - & 1.0000 & 0.5400 & 0.6494 & 0.7874 & 1.0000 & 0.8100 & 0.5525 & 0.7117 \\
 & 3.0 & - & 1.0000 & 0.5600 & 0.6410 & 0.7812 & 1.0000 & 0.8400 & 0.5435 & 0.7042 \\
 & 6.0 & - & 1.0000 & 0.4700 & 0.6803 & 0.8097 & 1.0000 & 0.8300 & 0.5464 & 0.7067 \\
\midrule
\cite{qu2025provably} 
 & 1.5 & - & 0.9700 & 0.9750 & 0.4987 & 0.6587 & 0.9850 & 0.9400 & 0.5117 & 0.6735 \\
 & 2.0 & - & 0.9600 & 0.9600 & 0.5000 & 0.6575 & 0.9800 & 0.9450 & 0.5091 & 0.6701 \\
 & 3.0 & - & 0.9700 & 0.9300 & 0.5105 & 0.6690 & 0.9650 & 0.9450 & 0.5052 & 0.6632 \\
 & 6.0 & - & 0.9850 & 0.9350 & 0.5130 & 0.6747 & 0.9650 & 0.9500 & 0.5039 & 0.6621 \\
\midrule
BREW (Ours) & 1.5
 & 0 & 0.1500 & 0.0250 & 0.8571 & 0.2553 & 0.1450 & 0.0250 & 0.8529 & 0.2479 \\
 & & 1 & 0.1650 & 0.0400 & 0.8049 & 0.2739 & 0.3050 & 0.0550 & 0.8472 & 0.4485 \\
 & & 3 & 0.3050 & 0.0500 & 0.8592 & 0.4502 & 0.4550 & 0.1550 & 0.7459 & 0.5652 \\
 & & 5 & 0.3600 & 0.0850 & 0.8090 & 0.4983 & 0.4600 & 0.3100 & 0.5974 & 0.5198 \\
\cmidrule(lr){2-11}
 & 2.0
 & 0 & 0.1850 & 0.0200 & 0.9024 & 0.3071 & 0.2450 & 0.0100 & 0.9608 & 0.3904 \\
 & & 1 & 0.3100 & 0.0550 & 0.8493 & 0.4542 & 0.3700 & 0.0850 & 0.8132 & 0.5086 \\
 & & 3 & 0.3800 & 0.0800 & 0.8261 & 0.5205 & 0.5800 & 0.1600 & 0.7838 & 0.6667 \\
 & & 5 & 0.4800 & 0.1150 & 0.8067 & 0.6019 & 0.6450 & 0.3050 & 0.6789 & 0.6615 \\
\cmidrule(lr){2-11}
 & 3.0
 & 0 & 0.3250 & 0.0100 & 0.9701 & 0.4869 & 0.3550 & 0.0150 & 0.9595 & 0.5182 \\
 & & 1 & 0.5050 & 0.0200 & 0.9619 & 0.6623 & 0.5300 & 0.0750 & 0.8760 & 0.6604 \\
 & & 3 & 0.6450 & 0.0700 & 0.9021 & 0.7522 & 0.7550 & 0.1350 & 0.8483 & 0.7989 \\
 & & 5 & 0.7100 & 0.1100 & 0.8659 & 0.7802 & 0.8650 & 0.2600 & 0.7689 & 0.8141 \\
\cmidrule(lr){2-11}
 & 6.0
 & 0 & 0.3550 & 0.0050 & 0.9861 & 0.5221 & 0.4300 & 0.0150 & 0.9663 & 0.5952 \\
 & & 1 & 0.5000 & 0.0350 & 0.9346 & 0.6515 & 0.6300 & 0.0750 & 0.8936 & 0.7390 \\
 & & 3 & 0.7300 & 0.0500 & 0.9359 & 0.8202 & 0.8150 & 0.1050 & 0.8859 & 0.8490 \\
 & & 5 & 0.8150 & 0.1100 & 0.8810 & 0.8468 & 0.9600 & 0.2650 & 0.7837 & 0.8629 \\
\bottomrule
\end{tabular}
\end{table}

\begin{table}[t]
\centering
\caption{Detection performance under 5\% token-increasing synonym substitution using OPT-1.3B on the OpenGen dataset.}
\label{tab:opt13b_opengen_synonym_insertion_5}
\renewcommand{\arraystretch}{1.2}
\begin{tabular}{l l l cccc cccc}
\toprule
\multicolumn{3}{c|}{\textbf{Setting}} & \multicolumn{4}{c|}{\textbf{T200}} & \multicolumn{4}{c}{\textbf{T500}} \\
\cmidrule(lr){1-3}\cmidrule(lr){4-7}\cmidrule(lr){8-11}
\textbf{Model} & $\delta$ & $s_{\max}$ & TPR & FPR & Precision & F1 & TPR & FPR & Precision & F1 \\
\midrule
MPAC
 & 1.5 & - & 1.0000 & 0.6050 & 0.6231 & 0.7678 & 1.0000 & 0.8900 & 0.5291 & 0.6920 \\
 & 2.0 & - & 1.0000 & 0.6250 & 0.6154 & 0.7619 & 1.0000 & 0.8850 & 0.5305 & 0.6932 \\
 & 3.0 & - & 1.0000 & 0.5800 & 0.6329 & 0.7752 & 1.0000 & 0.8650 & 0.5362 & 0.6981 \\
 & 6.0 & - & 0.9950 & 0.5800 & 0.6317 & 0.7728 & 1.0000 & 0.8850 & 0.5305 & 0.6932 \\
\midrule
\cite{qu2025provably}
 & 1.5 & - & 0.9500 & 0.9500 & 0.5000 & 0.6551 & 0.9850 & 0.9600 & 0.5064 & 0.6689 \\
 & 2.0 & - & 0.9750 & 0.9500 & 0.5064 & 0.6666 & 1.0000 & 0.9250 & 0.5194 & 0.6837 \\
 & 3.0 & - & 0.9800 & 0.9250 & 0.5144 & 0.6746 & 1.0000 & 0.9300 & 0.5181 & 0.6825 \\
 & 6.0 & - & 0.9600 & 0.9500 & 0.5026 & 0.6597 & 1.0000 & 0.9300 & 0.5181 & 0.6825 \\
\midrule
BREW (Ours) & 1.5
 & 0 & 0.3050 & 0.0200 & 0.9384 & 0.4603 & 0.3200 & 0.0250 & 0.9275 & 0.4758 \\
 & & 1 & 0.3850 & 0.0550 & 0.8750 & 0.5347 & 0.3550 & 0.0600 & 0.8554 & 0.5017 \\
 & & 3 & 0.5150 & 0.0550 & 0.9035 & 0.6561 & 0.5600 & 0.1400 & 0.8000 & 0.6588 \\
 & & 5 & 0.5650 & 0.0750 & 0.8828 & 0.6890 & 0.6900 & 0.2000 & 0.7752 & 0.7301 \\
\cmidrule(lr){2-11}
 & 2.0
 & 0 & 0.4600 & 0.0150 & 0.9684 & 0.6237 & 0.4500 & 0.0350 & 0.9278 & 0.6061 \\
 & & 1 & 0.5400 & 0.0250 & 0.9557 & 0.6901 & 0.5700 & 0.0500 & 0.9193 & 0.7037 \\
 & & 3 & 0.7100 & 0.0650 & 0.9161 & 0.8000 & 0.6950 & 0.1600 & 0.8128 & 0.7493 \\
 & & 5 & 0.8100 & 0.0650 & 0.9257 & 0.8640 & 0.8650 & 0.2350 & 0.7863 & 0.8238 \\
\cmidrule(lr){2-11}
 & 3.0
 & 0 & 0.6100 & 0.0250 & 0.9606 & 0.7462 & 0.6100 & 0.0300 & 0.9531 & 0.7439 \\
 & & 1 & 0.6300 & 0.0250 & 0.9618 & 0.7613 & 0.7000 & 0.0500 & 0.9333 & 0.8000 \\
 & & 3 & 0.8250 & 0.0350 & 0.9593 & 0.8871 & 0.8850 & 0.1550 & 0.8510 & 0.8676 \\
 & & 5 & 0.9200 & 0.1100 & 0.8932 & 0.9064 & 0.9400 & 0.2700 & 0.7768 & 0.8506 \\
\cmidrule(lr){2-11}
 & 6.0
 & 0 & 0.5900 & 0.0100 & 0.9833 & 0.7375 & 0.7900 & 0.0750 & 0.9132 & 0.8471 \\
 & & 1 & 0.6300 & 0.0100 & 0.9843 & 0.7682 & 0.8200 & 0.1050 & 0.8864 & 0.8519 \\
 & & 3 & 0.5250 & 0.0550 & 0.9050 & 0.6640 & 0.9300 & 0.1300 & 0.8773 & 0.9029 \\
 & & 5 & 0.9450 & 0.0800 & 0.9219 & 0.9333 & 0.9750 & 0.2000 & 0.8297 & 0.8965 \\
\bottomrule
\end{tabular}
\end{table}

\begin{table}[t]
\centering
\caption{Detection performance under 10\% token-increasing synonym substitution using OPT-1.3B on the OpenGen dataset.}
\label{tab:opt13b_opengen_synonym_insertion_10}
\renewcommand{\arraystretch}{1.2}
\begin{tabular}{l l l cccc cccc}
\toprule
\multicolumn{3}{c|}{\textbf{Setting}} & \multicolumn{4}{c|}{\textbf{T200}} & \multicolumn{4}{c}{\textbf{T500}} \\
\cmidrule(lr){1-3}\cmidrule(lr){4-7}\cmidrule(lr){8-11}
\textbf{Model} & $\delta$ & $s_{\max}$ & TPR & FPR & Precision & F1 & TPR & FPR & Precision & F1 \\
\midrule
MPAC
 & 1.5 & - & 0.9800 & 0.6200 & 0.6125 & 0.7538 & 1.0000 & 0.9150 & 0.5222 & 0.6861 \\
 & 2.0 & - & 0.9900 & 0.5950 & 0.6246 & 0.7660 & 1.0000 & 0.8800 & 0.5319 & 0.6944 \\
 & 3.0 & - & 1.0000 & 0.6400 & 0.6098 & 0.7576 & 1.0000 & 0.8950 & 0.5277 & 0.6908 \\
 & 6.0 & - & 1.0000 & 0.6550 & 0.6042 & 0.7533 & 1.0000 & 0.8900 & 0.5291 & 0.6920 \\
\midrule
\cite{qu2025provably} 
 & 1.5 & - & 0.9400 & 0.9200 & 0.5053 & 0.6573 & 0.9900 & 0.9450 & 0.5116 & 0.6746 \\
 & 2.0 & - & 0.9550 & 0.9350 & 0.5052 & 0.6608 & 0.9600 & 0.9350 & 0.5065 & 0.6632 \\
 & 3.0 & - & 0.9850 & 0.9350 & 0.5130 & 0.6746 & 1.0000 & 0.9600 & 0.5102 & 0.6756 \\
 & 6.0 & - & 0.9500 & 0.9500 & 0.5000 & 0.6551 & 1.0000 & 0.9400 & 0.5154 & 0.6802 \\
\midrule
BREW (Ours) & 1.5
 & 0 & 0.1100 & 0.0050 & 0.9565 & 0.1973 & 0.1450 & 0.0400 & 0.7837 & 0.2447 \\
 & & 1 & 0.1850 & 0.0400 & 0.8222 & 0.3020 & 0.2850 & 0.0600 & 0.8261 & 0.4237 \\
 & & 3 & 0.1950 & 0.0550 & 0.7800 & 0.3120 & 0.4450 & 0.1650 & 0.7295 & 0.5527 \\
 & & 5 & 0.4050 & 0.0900 & 0.8181 & 0.5418 & 0.5500 & 0.2300 & 0.7051 & 0.6179 \\
\cmidrule(lr){2-11}
 & 2.0
 & 0 & 0.2950 & 0.0200 & 0.9365 & 0.4487 & 0.2400 & 0.0300 & 0.8889 & 0.3780 \\
 & & 1 & 0.2250 & 0.0200 & 0.9183 & 0.3614 & 0.3750 & 0.0600 & 0.8621 & 0.5226 \\
 & & 3 & 0.4500 & 0.0800 & 0.8491 & 0.5882 & 0.5850 & 0.1350 & 0.8125 & 0.6802 \\
 & & 5 & 0.4900 & 0.0800 & 0.8596 & 0.6242 & 0.6750 & 0.2600 & 0.7219 & 0.6976 \\
\cmidrule(lr){2-11}
 & 3.0
 & 0 & 0.2450 & 0.0050 & 0.9800 & 0.3920 & 0.3300 & 0.0250 & 0.9296 & 0.4871 \\
 & & 1 & 0.3900 & 0.0250 & 0.9397 & 0.5512 & 0.4900 & 0.0750 & 0.8672 & 0.6261 \\
 & & 3 & 0.6000 & 0.0800 & 0.8823 & 0.7142 & 0.7900 & 0.1200 & 0.8681 & 0.8272 \\
 & & 5 & 0.6650 & 0.0900 & 0.8807 & 0.7578 & 0.8550 & 0.2150 & 0.7991 & 0.8261 \\
\cmidrule(lr){2-11}
 & 6.0
 & 0 & 0.4100 & 0.0050 & 0.9879 & 0.5795 & 0.4600 & 0.0250 & 0.9484 & 0.6195 \\
 & & 1 & 0.4800 & 0.0200 & 0.9600 & 0.6400 & 0.6850 & 0.0850 & 0.8896 & 0.7740 \\
 & & 3 & 0.7100 & 0.0650 & 0.9161 & 0.8000 & 0.8200 & 0.1400 & 0.8541 & 0.8367 \\
 & & 5 & 0.7450 & 0.0850 & 0.8975 & 0.8142 & 0.9450 & 0.2150 & 0.8146 & 0.8750 \\
\bottomrule
\end{tabular}
\end{table}

\subsection{Detailed Results under Paraphrasing Attacks}
\label{Appendix:Paraphrasing_attacks}

Table~\ref{tab:Paraphrasing_attacks} reports detailed detection results under paraphrasing attacks generated by the T5-based paraphrasing model \cite{raffel2020exploring}. All experiments are conducted on texts truncated to $T=200$ tokens using OPT-1.3B, evaluated on both the C4 and OpenGen datasets.

The table provides a fine-grained view of the performance trends summarized by the ROC curves in the main text. MPAC consistently achieves high true positive rates (TPR), but exhibits moderately elevated false positive rates (FPR), indicating partial robustness to paraphrasing but limited reliability in low-FPR regimes. \cite{qu2025provably} further amplifies this trade-off, attaining near-saturated TPR at the cost of extremely high FPR, which results in near-random discrimination and poor practical usability.

In contrast, BREW maintains strict control over FPR across all configurations of $\delta$ and window-shift range $s_{\max}$, while gradually improving TPR as $s_{\max}$ increases. This stable TPR--FPR balance under strong semantic rewriting explains the superior ROC behavior of BREW observed in the main text.

\begin{table}[t]
\centering
\caption{Detailed detection performance under paraphrasing attacks generated by the T5-based paraphrasing model on C4 and OpenGen ($T=200$, OPT-1.3B). MPAC achieves high TPR but incurs moderately elevated FPR, while \cite{qu2025provably} attains similarly high TPR with near-saturated FPR, leading to unreliable discrimination. In contrast, BREW consistently maintains low FPR across all $\delta$ and $s_{\max}$ settings, while steadily improving TPR as the window-shift range increases.}
\label{tab:Paraphrasing_attacks}
\renewcommand{\arraystretch}{1.2}
\begin{tabular}{l cc | cccc | cccc}
\toprule
\multicolumn{3}{c|}{\textbf{Setting}} & \multicolumn{8}{c}{\textbf{T200}} \\
\cmidrule(lr){1-3}\cmidrule(lr){4-11}
\multicolumn{3}{c|}{\textbf{Dataset}} & \multicolumn{4}{c|}{\textbf{C4}} & \multicolumn{4}{c}{\textbf{OpenGen}} \\
\cmidrule(lr){1-3}\cmidrule(lr){4-7}\cmidrule(lr){8-11}
\textbf{Model} & $\delta$ & $s_{\max}$ & \textbf{TPR} & \textbf{FPR} & \textbf{Precision} & \textbf{F1} & \textbf{TPR} & \textbf{FPR} & \textbf{Precision} & \textbf{F1} \\
\midrule
 MPAC
& 1.5 & - & 0.950 & 0.590 & 0.6169 & 0.7480 & 0.930 & 0.600 & 0.6078 & 0.7352  \\
& 2 & - & 0.990 & 0.560 & 0.6387 & 0.7765  & 1.000 & 0.660 & 0.6024 & 0.7519 \\
& 3 & - & 1.000  & 0.580 & 0.6329 & 0.7752 & 1.000 & 0.630 & 0.6135 & 0.7605 \\
& 6 & - & 1.000  & 0.530 & 0.6536 & 0.7905  & 1.000  & 0.570 & 0.6369 & 0.7782 \\
\midrule
 \cite{qu2025provably} 
& 1.5 & - & 0.950 & 0.950 & 0.5000 & 0.6552 & 0.920 & 0.960 & 0.4894 & 0.6389 \\
& 2 & - & 0.920 & 0.930 & 0.4973 & 0.6456 & 0.960 & 0.960 & 0.5000 & 0.6575 \\
& 3 & - & 0.960 & 0.970 & 0.4974 & 0.6553 & 0.960 & 0.920 & 0.5106 & 0.6667 \\
& 6 & - & 0.980 & 0.900 & 0.5213 & 0.6806 & 0.980 & 0.930 & 0.5131 & 0.6735 \\
\midrule
 BREW (Ours)
& 1.5 
& 0 & 0.360 & 0.020 & 0.9474 & 0.5217 & 0.450 & 0.020 & 0.9574 & 0.6122 \\
& & 1 & 0.410 & 0.040 & 0.9111 & 0.5655 & 0.400 & 0.000 & 1.0000 & 0.5714 \\
& & 3 & 0.350 & 0.030 & 0.9211 & 0.5072 & 0.570 & 0.050 & 0.9194 & 0.7037 \\
& & 5 & 0.480 & 0.100 & 0.8276 & 0.6076 & 0.480 & 0.100 & 0.8276 & 0.6076 \\
\cmidrule(lr){2-11}
& 2 
& 0 & 0.330 & 0.000 & 1.0000 & 0.4962 & 0.520 & 0.010 & 0.9811 & 0.6797 \\
& & 1 & 0.530 & 0.030 & 0.9464 & 0.6795 & 0.540 & 0.020 & 0.9643 & 0.6923 \\
& & 3 & 0.560 & 0.060 & 0.9032 & 0.6914 & 0.710 & 0.030 & 0.9595 & 0.8161 \\
& & 5 & 0.580 & 0.140 & 0.8056 & 0.6744 & 0.530 & 0.130 & 0.8030 & 0.6386 \\
\cmidrule(lr){2-11}
& 3 
& 0 & 0.360 & 0.000 & 1.0000 & 0.5294 & 0.600 & 0.000 & 1.0000 & 0.7500 \\
& & 1 & 0.450 & 0.040 & 0.9184 & 0.6040 & 0.630 & 0.060 & 0.9130 & 0.7456 \\
& & 3 & 0.650 & 0.040 & 0.9420 & 0.7692 & 0.710 & 0.060 & 0.9221 & 0.8023 \\
& & 5 & 0.680 & 0.110 & 0.8608 & 0.7598 & 0.720 & 0.040 & 0.9474 & 0.8182 \\
\cmidrule(lr){2-11}
& 6 
& 0 & 0.550 & 0.000 & 1.0000 & 0.7097 & 0.550 & 0.010 & 0.9821 & 0.7051 \\
& & 1 & 0.720 & 0.030 & 0.9600 & 0.8229 & 0.780 & 0.080 & 0.9070 & 0.8387 \\
& & 3 & 0.710 & 0.100 & 0.8765 & 0.7845 & 0.700 & 0.070 & 0.9091 & 0.7910 \\
& & 5 & 0.810 & 0.120 & 0.8710 & 0.8394 & 0.910 & 0.070 & 0.9286 & 0.9192 \\
\bottomrule
\end{tabular}
\end{table}


\end{document}